\documentclass[a4paper,UKenglish,cleveref, autoref, thm-restate, numberwithinsect]{lipics-v2021}

\hideLIPIcs  

\usepackage{amsmath,stmaryrd}
\usepackage{mathtools}
\usepackage{amsthm}
\usepackage{url}
\usepackage{psfrag}
\usepackage{color}
\usepackage{hyperref,breakurl}
\usepackage{breakurl}
\usepackage[all]{xy}
\usepackage{mathpartir}

\usepackage{epsfig}
\usepackage{ushort} 

\newcommand{\Comment}[1]{}

\newcommand{\il}[1]{{\color{red}{#1}}}
\newcommand{\cm}[1]{{\color{blue}{#1}}}

\newcommand{\red}{\to}

\newcommand{\tran}[1]{\xrightarrow[]{#1}} 
\newcommand{\rtran}[1]{\stackrel{#1}{\rightsquigarrow}} 

\newcommand{\ftran}[1]{\stackrel{#1}{\rightharpoonup}} 
\newcommand{\rectran}[1]{\stackrel{#1}{\rightarrow}_\mathsf{r}} 
\newcommand{\wrectran}[1]{\stackrel{#1}{\Rightarrow}_\mathsf{r}} 
\newcommand{\recred}{\rightarrow_\mathsf{r}} 
\newcommand{\wrecred}{\Rightarrow_\mathsf{r}} 

\newcommand{\Proc}{\mathsf{Proc}}

\newcommand{\key}[1]{[#1]}

\newcommand{\fk}[1]{\mathsf{fk}(#1)}

\newcommand{\std}[1]{\mathsf{std}(#1)}
\newcommand{\tostd}[1]{\mathsf{toStd}(#1)}

\newcommand{\Par}{\mid}

\newcommand{\rev}[1]{\ushort{#1}} 

\newcommand{\nil}{\mathbf{0}}

\def\finex{{\unskip\nobreak\hfil
\penalty50\hskip1em\null\nobreak\hfil$\diamond$
\parfillskip=0pt\finalhyphendemerits=0\endgraf}}

\newcommand{\inter}{\cap}

\definecolor{darkgreen}{rgb}{0.0, 0.5, 0.0}
\definecolor{applegreen}{rgb}{0.55, 0.71, 0.0}
\definecolor{amber(sae/ece)}{rgb}{1.0, 0.49, 0.0}
\definecolor{airforceblue}{rgb}{0.36, 0.54, 0.66}
\definecolor{amethyst}{rgb}{0.6, 0.4, 0.8}
\definecolor{brickred}{rgb}{0.8, 0.25, 0.33}
		
\newcommand{\irek}[1]{{\color{darkgreen}{#1}}}

\newcommand{\Act}{\mathsf{Act}}

\mathchardef\mh="2D

\newcommand{\enc}[1]{[\![ #1 ]\!]}

\newcommand{\encpref}[1]{\enc{ #1 }}
\newcommand\smallbullet{%
    \raisebox{-0.8ex}{\scalebox{1.5}{$\cdot$}}%
}
\newcommand{\attr}[1]{#1 \smallbullet R}

\newcommand{\keys}[1]{\mathsf{k}(#1)}

\newcommand{\Co}[1]{\overline{#1}}

\newcommand{\out}[1]{\overline{#1}}

\newcommand{\rec}[2]{\mathtt{rec}\ #1.#2}
\newcommand{\pair}[2]{(#1\ ,\  #2)}
\newcommand{\tup}[1]{\langle #1 \rangle}

\newcommand{\rcal}{{\mathcal R}}
\newcommand{\bisim}{\sim}

\newcommand{\wpitran}[1]{\xRightarrow[\raisebox{1pt}{\makebox[1ex]{\ensuremath{\scriptstyle\pi}}}]{#1}}
\newcommand{\pitran}[1]{\xrightarrow[\raisebox{1pt}{\makebox[1ex]{\ensuremath{\scriptstyle\pi}}}]{#1}}

\newcommand{\pilab}{\lambda} 
\newcommand{\pistr}{\equiv_{\pi}} 
\newcommand{\str}{\equiv} 
\newcommand{\phiel}[3]{(#1,#2) \leftrightarrow #3} 

\newcommand{\hol}{\bullet}
\newcommand{\tick}{\checkmark}
\newcommand{\barb}{\downarrow}
\newcommand{\wbarb}{\Downarrow}
\newcommand{\wtran}[1]{\stackrel{#1}{\Rightarrow}}
\newcommand{\wred}{\Rightarrow}
\newcommand{\presim}[1]{\preccurlyeq^{#1}}
\newcommand{\recpresim}[1]{\preccurlyeq^{#1}_\mathsf{r}}
\newcommand{\calc}{\mathbb{C}} 
\newcommand{\opr}{\mathsf{op}} 

\newcommand{\subst}[2]{\{#1/#2\}}
\newcommand{\fn}{\mathrm{fn}} 
\newcommand{\bn}{\mathrm{bn}}
\newcommand{\n}{\mathrm{n}}

\newcommand{\tree}{\mathtt{tr}}

\title{On the Encodability of Reversible Process Calculi} 

\author{Ivan Lanese}{Olas Team, University of Bologna/INRIA, Italy \and \url{https://www.unibo.it/sitoweb/ivan.lanese}}{ivan.lanese@unibo.it}{https://orcid.org/0000-0003-2527-9995}{partial support of the French ANR project SmartCloud ANR-23-CE25-0012.}

\author{{Claudio Antares} Mezzina}{University of Bari Aldo Moro, Italy \and \url{https://cmezzina.github.io/}}{cmezzina@gmail.com}{https://orcid.org/0000-0003-1556-2623}{partial support of the MSCA SE project QCOMICAL (Grant Agreement ID: 101182520) and of the Japan Society for the Promotion of Science (JSPS) through the JSPS Invitation Fellowship for Research in Japan, grant no. S25016.}

\author{Iain Phillips}{Imperial College London, UK \and \url{https://www.doc.ic.ac.uk/~iccp}}{i.phillips@imperial.ac.uk}{https://orcid.org/0000-0001-5013-5876}{}

\author{Irek Ulidowski}{University of Leicester, UK; AGH University of Krak\'{o}w, Poland \and \url{https://www.cs.le.ac.uk/people/iulidowski/}}{iu3@leicester.ac.uk}{https://orcid.org/0000-0002-3834-2036}{partial support of 
the AY2024 International PI Invitation Program, IAR Nagoya University, and of
the Japan Society for the Promotion of Science (JSPS) through the JSPS Invitation Fellowship for Research in Japan, grant no. S21050.}

\author{Shoji Yuen}{Nagoya University, Japan \and \url{https://profs.provost.nagoya-u.ac.jp/html/100001809_en.html}}{yuen@sqlab.jp}{https://orcid.org/0000-0003-2642-0647}{}

\authorrunning{I.~Lanese, C.~A.~Mezzina, I.~Phillips, I.~Ulidowski, S.~Yuen} 

\Copyright{I.~Lanese, C.~A.~Mezzina, I.~Phillips, I.~Ulidowski, S.~Yuen} 



\nolinenumbers
\begin{document}

\maketitle

\begin{abstract}
Reversibility, allowing one to execute a program not only forwards as usual, but also backwards, has emerged as a fundamental concept in computing, with applications ranging from debugging and fault tolerance to biological and quantum systems. CCSK, a reversible extension of CCS, is a paradigmatic model of reversible concurrent computation.
  In this paper, we investigate the encodability of CCSK into classical forward-only concurrent models. We establish a separation theorem showing that there is no basic, success-sensitive encoding of CCSK into CCS or the $\pi$-calculus, highlighting the strong impact of reversibility on expressive power. We then present an encoding of CCSK processes with only top-level parallel composition into the
internal $\pi$-calculus, correct up to strong bisimilarity.  
We also identify a fundamental limitation: no parallel-preserving encoding of CCSK (with arbitrary parallel composition) into the $\pi$-calculus can be correct up to strong bisimilarity. 
 Finally, we provide a parallel-preserving encoding correct under a weaker behavioural correspondence: weak mutual simulation.
  Our
  findings extend the literature of encodability results to reversible process calculi.
\end{abstract}

\section{Introduction}

Reversibility has emerged as a fundamental concept in the theory of computation. In concurrent systems, reversibility requires that every action can be undone in a causally consistent way~\cite{DK04}, thereby restoring past states without breaking causal dependencies. This paradigm has proved useful in debugging~\cite{microsoft,GiachinoLM14,LaneseSUS22}, where reversing erroneous computations is more natural than replaying from scratch; in biochemical modelling~\cite{PhillipsUY12,MelgrattiMP22,KuhnU22}, 
given that many biochemical reactions are inherently reversible;
in fault tolerance, where recovery requires controlled rollback to consistent states~\cite{VassorS18,MezzinaTY25}; and in quantum computing, where reversibility is
forced by the laws of physics.

Several reversible extensions of process calculi have been proposed \cite{DK04,PU07,LaneseMS16,CristescuKV13}. Among them, CCSK~\cite{PU06,PU07} enriches CCS with keys that record causal dependencies, supporting both forward and backward execution steps. Its operational semantics preserves the familiar flavour of CCS while enabling precise reversal of computations, including synchronisations. 

A natural and fundamental question arises:
\textit{can reversibility be encoded into classical forward process calculi}?

The expressive power of forward process calculi has been extensively
studied. For instance, Palamidessi showed that the
synchronous $\pi$-calculus cannot be encoded into its asynchronous
counterpart~\cite{Palamidessi03}. Gorla developed a uniform framework for separation and encodability results, identifying robust criteria for valid encodings~\cite{Gor10}. More recently, Pugliese and Tiezzi introduced \emph{replacement freeness} as a general criterion for proving separation results between calculi~\cite{PT20}. Further information on encodings can be found in a survey on the topic~\cite{Peters19}. However, the role of reversibility in this landscape remains poorly understood. Can reversible calculi such as CCSK be faithfully encoded in forward ones like CCS or the $\pi$-calculus?

While broadly speaking the answer should be positive, as shown, e.g., by the encoding in~\cite{LaneseMS16},
our deeper investigation
shows a number
of subtleties that need to be considered in order to define such
encodings, and that encoding reversibility is quite demanding in terms
of the requirements on, e.g., the target calculus.

In more detail, this paper gives the first systematic account of the
expressive power of CCSK. Our contributions are as follows:
\begin{description}
\item[Separation theorem:] We prove (Theorem~\ref{thm:basic succ sens CCSK pi}) that there exists no basic, suc\-cess-sensitive encoding of CCSK into CCS or the $\pi$-calculus (with recursion and without matching). The proof adapts techniques from  Pugliese-Tiezzi~\cite{PT20} and Gorla~\cite{Gor10}, and relies only on a fragment of CCSK (nil, prefixing, and top-level parallel composition), showing that the result is robust. This establishes that reversibility in CCS is not eliminable: it
requires expressive power which goes beyond that provided by classical CCS.

\item[Encoding into $\pi$:] We define (Figure~\ref{fig:enc_strong}) a parallel preserving encoding (Definition~\ref{def:nary par pres enc}) of CCSK processes with top-level parallel composition only into the $\pi$-calculus with internal mobility~\cite{Sangiorgi96a}. The encoding exploits $\pi$ names to represent CCSK keys and it is correct up to strong bisimilarity. This demonstrates how far CCSK can be captured within $\pi$, while making precise where the correspondence breaks down.

\item[Limits of behavioural correspondence:]
We show (Theorem~\ref{thm:par pres strong bis}) that no encoding of CCSK into the $\pi$-calculus that is
parallel preserving
can be correct up to strong bisimilarity. This identifies a sharp boundary: top-level parallelism can be handled, but general parallelism exposes an inherent mismatch between reversible and forward concurrency.
\item[Encoding general parallel composition:] we refine (Definition~\ref{def:tree}) the encoding of CCSK into the $\pi$-calculus mentioned above to cope not only with top-level parallel composition, but also with parallel composition at lower levels. This requires the definition of a multiparty protocol to ensure that all descendants of a process have reversed before reversing it; hence the encoding is correct only up to \emph{weak} mutual simulation.
\end{description}
Our results situate CCSK (and in general reversibility in CCS) in the landscape of
encodability results
between process calculi. They reveal that reversibility
requires
expressive power not available in classical (forward)
CCS,
for two main conceptual reasons:
(i) the encoding cannot be compositional since the encoding of a term depends on its past, as it can go back to such past;
(ii) the encoding of parallel composition in a calculus with binary synchronisation such as CCS or the $\pi$-calculus requires a complex protocol to ensure all the children of a process are back to their initial state in order to enable backward execution of the parent process.
These results are not dependent on features which are specific to CCSK, but on general aspects of reversibility; hence, similar results should hold for other reversible calculi as well. While we give some indications in this direction,
we leave a more detailed analysis of
this topic
for future work.


Along with the theoretical interpretation of our results, the encoding techniques suggest systematic ways of translating causal histories into explicit control structures, which could inspire runtime mechanisms for reversible execution. In particular, using communication keys as explicit channels in the $\pi$-calculus highlights how one might represent rollback information as first-class messages, paving the way for implementations in distributed 
settings. 

\noindent{\bf Structure of the paper}: Sections~\ref{sec:CCSK} and~\ref{sec:pi} 
briefly recall
CCSK and the $\pi$-calculus. Section~\ref{sec:separation} presents a separation result. In Section~\ref{sec:encoding} we present
an encoding of CCSK into the $\pi$-calculus with top-level parallel composition only, correct up to strong bisimilarity. 
Section~\ref{sec:weak_encoding}  handles lower-level parallel composition, using the weaker notion of mutual simulation.
%
The final section discusses related work and concludes the paper. 

\section{CCSK: syntax and operational semantics}\label{sec:CCSK}
In this section we recall the syntax and semantics of CCSK~\cite{PU06}.

Let $\mathcal N = \{a, b, c, \ldots\}$ be a set of \emph{names}
(also called channels), 
and let $\overline{\mathcal{N}} = \{\overline{a} \mid a \in \mathcal{N}\}$ be the set
of their corresponding \emph{co-names}. The set of all \emph{actions} is
$\Act = \mathcal{N} \cup \overline{\mathcal{N}} \cup \{\tau\}$, where $\mathcal{N}$ and $\overline{\mathcal{N}}$ contain, respectively, input and output actions, and
$\tau$ denotes the silent action. We let $\alpha,\beta$ range over the set
$\mathcal{N} \cup \overline{\mathcal{N}}$, while $\mu$ ranges over the set $\Act$.
We say two prefixes such as $a$ and $\overline{a}$ are \emph{complementary}.
The syntax of CCSK is given below. 
The set of CCSK terms is $\Proc$, and we shall refer to terms as \emph{processes}. 
We let $P, Q$ and their primed or subscripted versions range over processes.
\[
\begin{array}{lll}
P & ::= \Sigma_I \rho_i.P_i \mid P\Par Q \ \mid \ \nu a.P \qquad 
\rho &::= \mu \mid \mu[k]
\end{array}
\]
We use $n$-ary guarded choice $\Sigma_I \rho_i.P_i$, indexed by a set~$I$, to denote a process which may execute (or may already have executed, see below) any of the prefixes $\rho_i$ and then behave as~$P_i$. 
This will 
simplify our encoding while syntactically ensuring that choice is always guarded. 
We assume
$\Sigma_I \rho_i.P_i = \nil$ (the nil process, which does nothing) when $I = \emptyset$. We may write
just $\rho.P$ when $I$ is a singleton, and $\rho_1.P_1+\Sigma_{I \setminus \{1\}} \rho_i.P_i$ to highlight a single branch, assuming guarded choice to be associative and commutative.
Prefix $\rho$ can be either $\mu$ (representing an action to be executed) or $\mu[k]$ (representing an action which has already been executed). 
Here $k$ is a \emph{key}, namely an identifier for the action execution. We denote with $\mathcal{K}$ the set of all keys.
Synchronising actions have the same key, to ensure that they are undone together. This will become clearer when discussing the operational semantics.
Process $P \Par Q$ is the parallel composition of $P$ and $Q$, and name $a$ is bound in the restriction $\nu a.P$.

\begin{definition}[Process keys]\label{def:keys}
The set of keys of a CCSK process $P$, written $\keys{P}$, is inductively defined as follows:
\begin{align*}
&\keys{\mu.P} =
\keys{\nil} = \emptyset && \keys{\mu\key{k}.P} = \{k\}\cup \keys{P}
&& \keys{P \Par Q} = \keys{P} \cup \keys{Q} \\
& \keys{\Sigma_I \rho_i.P_i} = \bigcup_I \keys{\rho_i.P_i}
&& \keys{\nu a.P} = \keys{P}
\end{align*}
\end{definition}

\emph{Contexts} $C$ are obtained by extending with a hole $\hol$ the syntax of processes. We consider contexts with a single occurrence of~$\hol$. We write $C[P]$ for the process obtained by replacing the $\bullet$ in $C$ with process $P$.


We define $\str$ as the smallest congruence relation closed under the 
rules of Figure~\ref{fig:ccsk_sem},
where $P =_{\alpha} Q$ indicates two processes equivalent modulo $\alpha$-conversion of bound names.
\begin{figure}[t]
\[
\begin{array}{c}
P \Par \nil \str P \;\quad P \Par Q \str Q \Par P \;\quad   P_1 \Par (P_2 \Par P_3) \str (P_1 \Par P_2) \Par P_3 \;\quad
P =_{\alpha} Q \implies P \str Q
\end{array}
\]
	\caption{CCSK structural congruence rules. 
	}
	\label{fig:ccsk_sem}
\end{figure}
The forward semantics of CCSK is given by the labelled 
transition system $(\Proc, \Act \times \mathcal{K},\ftran{})$, where
${\ftran{}} \subseteq \Proc \times (\Act \times \mathcal{K}) \times \Proc$
is the smallest relation closed under the forward rules of Figure~\ref{fig:ccsk_sem_full},
and under structural congruence.
%
\begin{figure}[tbh!]
	\begin{mathpar}
	\inferrule*[lab={\scriptsize(\textsc{Act1})}]{\std{Q}}{\mu.Q \ftran{\mu[k]}\mu[k].Q} \and 
	\inferrule*[lab={\scriptsize(\textsc{Act1}$^{\bullet}$)}]{\std{Q}}{\mu[k].Q \rtran{\mu[k]}\mu.Q}\\
	\inferrule*[lab={\scriptsize(\textsc{Act2})}]{P\ftran{\mu[h]}P' \and h \neq k}{\alpha[k].P \ftran{\mu[h]}\alpha[k].P'} \and 
	\inferrule*[lab={\scriptsize(\textsc{Act2}$^{\bullet}$)}]{P'\rtran{\mu[h]}P \and h \neq k}{\alpha[k].P' \rtran{\mu[h]}\alpha[k].P} \\ 
%
%
	\inferrule*[lab={\scriptsize(\textsc{Sum})}]{P_i\ftran{\mu[k]}P_i' \quad \forall j\neq i .\ (P_j' = P_j \wedge \std{P_j}) }{\Sigma_I \,P_i \ftran{\mu[k]}\Sigma_I \,P'_i }\and
%
	\inferrule*[lab={\scriptsize(\textsc{Sum}$^{\bullet}$)}]{P'_i\rtran{\mu[k]}P_i \quad  \forall j\neq i .\ ( P_j = P_j' \wedge \std{P_j'} )}{\Sigma_I \,P'_i \rtran{\mu[h]}\Sigma_I \,P_i } 
		\\ 
	\inferrule*[lab={\scriptsize(\textsc{Par})}]{P \ftran{\mu[k]} P' \and k\not \in \keys{Q}}{P \Par Q \ftran{\mu[k] }P' \Par Q} \and 
	\inferrule*[lab={\scriptsize(\textsc{Par}$^\bullet$)}]{P' \rtran{\mu[k]} P \and k\not \in \keys{Q}}{P' \Par Q \rtran{\mu[k] }P \Par Q}\\
	\inferrule*[lab={\scriptsize(\textsc{Syn})}]{P \ftran{\alpha[k]} P' \and Q \ftran{\Co{\alpha}[k]} Q'}{P \Par Q \ftran{\tau[k] }P' \Par Q'} \and 
	\inferrule*[lab={\scriptsize(\textsc{Syn}$^\bullet$)}]{P' \rtran{\alpha[k]} P \and Q' \rtran{\Co{\alpha}[k]} Q}{P' \Par Q' \rtran{\tau[k] }P \Par Q} \\
	\inferrule*[lab={\scriptsize(\textsc{Res})}]{P \ftran{\mu[k]} P' \and \mu \not\in \{a,\Co{a}\}}{\nu a.P   \ftran{\mu[k] } \nu a.P'} 
				\and 
	\inferrule*[lab={\scriptsize(\textsc{Res}$^\bullet$)}]{P' \rtran{\mu[k]} P \and \mu \not\in \{a,\Co{a}\} }{\nu a.P'   \rtran{\mu[k] }\nu a.P } \\
	\inferrule*[lab={\scriptsize(\textsc{Str})}]{P \str P' \quad P' \ftran{\mu[k]} Q' \quad Q' \str Q }{P   \ftran{\mu[k] } Q} 
				\and 
	\inferrule*[lab={\scriptsize(\textsc{Str}$^\bullet$)}]{P \str P' \quad P' \rtran{\mu[k]} Q' \quad Q' \str Q }{P   \rtran{\mu[k] } Q } 
	\end{mathpar}
	\caption{CCSK labelled transition rules.} 
	\label{fig:ccsk_sem_full}
\end{figure}
Correspondingly, the backward semantics of a CCSK process is the
smallest relation $\rtran{\,}$ closed under the 
backward rules of Figure~\ref{fig:ccsk_sem_full}, which are symmetric to the forward ones.
The semantics of CCSK is the union of the two relations, denoted by $\tran{}$.
Notably, each rule executing some action is paired with a dual rule undoing the same action.
To this end, executed actions are not dropped from the process as in classical process calculi, but equipped with a key, thus denoting that they belong to the \emph{past} of the process. This is visible, e.g., in rule {\scriptsize(\textsc{Act1})}.
Due to this, execution can occur under executed prefixes, cf.~rule {\scriptsize(\textsc{Act2})}. Rule {\scriptsize(\textsc{Syn})} ensures that synchronising actions have the same key; thus rule {\scriptsize(\textsc{Syn}$^\bullet$)} is needed to ensure they are undone together. Note that rule {\scriptsize(\textsc{Par}$^\bullet$)} is not applicable in this case, due to the side condition $k\not \in \keys{Q}$.


\begin{definition}[Guarded and top-level sub-processes]\label{def:CCSK top-level}
Given a CCSK process $P$, a sub-process $Q$ is \emph{guarded}, or \emph{lower level}, in $P$ if it occurs within a summation $\sum_I\rho_i.P_i$; otherwise we say it is \emph{top level}.
Process $P$ has only top-level parallel composition if every sub-process of the form $Q_1 \Par Q_2$ is top level.
\end{definition}

\begin{definition}[Standard process]
A CCSK process $P$ is \emph{standard}, written $\std{P}$, if it contains no keys, that is $\keys{P}=\emptyset$.
\end{definition}

Standard processes are CCS processes.
Not all the processes generated by the grammar are meaningful. To
cope with this issue
we define what is a reachable process.
\begin{definition}[Reachable process]
A CCSK process $P$ is 
\emph{reachable} if it can be derived
from a standard process using the rules of Figure~\ref{fig:ccsk_sem_full}.
\end{definition}
From now on we will restrict attention to reachable processes.

Next we recall two definitions, which will come
in handy once we define our encoding in Section~\ref{sec:encoding}.




\Comment{We now recall two definitions for CCSK processes, which will come
in handy once we define our encoding in Section~\ref{sec:encoding}.
}

\begin{definition}[Free and bound keys]\label{def:fbk}
 A key $k$ is \emph{bound} in a process~$P$ iff
it occurs either twice, attached to complementary prefixes, or once, attached to
a $\tau$ prefix. A key $k$ is \emph{free} if it occurs once, attached to a non-$\tau$ prefix. We will indicate
with $\fk{P}$ the set of free keys of $P$.
\end{definition}

Thanks
to~\cite[Prop.\ 3.5]{LaneseP21}
keys are either free or bound.

\begin{definition}
Let $\tostd{\cdot}$ be a forgetful map on CCSK processes defined  as follows:
\begin{align*}
&\tostd{P} = P \text{ if } \std{P}  &&
\tostd{\mu\key{k}.P} = \mu.\tostd{P}\\
&\tostd{\Sigma_I \rho_i.P_i} = \Sigma_I \tostd{\rho_i.P_i}&&
\tostd{P \Par Q} = \tostd{P} \Par \tostd{Q} \\
&\tostd{\nu a.P} =\nu a.\tostd {P} 
\end{align*}
\end{definition}
Intuitively, function $\tostd{\cdot}$ removes all keys from a process, hence undoing all its actions.

%
\begin{example}[CCSK computation]
Consider the CCSK process $a.(b.\nil + c.\nil)$. Executing action $a$ then $b$ is represented by  transitions
\[
a.(b.\nil + c.\nil) \ftran{a[k]} a[k].(b.\nil + c.\nil) \ftran{b[h]} a[k].(b[h].\nil + c.\nil).
\]
Executed actions are not dropped from the terms, as in CCS, but decorated with keys. Also,
unused branches of the choice (here $c.\nil$) are kept, to allow their future execution. After $a,b$ have been done, 
we can undo the action $b[h]$ and then choose to explore $c$ instead:
\[
a[k].(b[h].\nil + c.\nil) \rtran{b[h]} a[k].(b.\nil + c.\nil)\ftran{c[l]} a[k].(b.\nil + c[l].\nil)
\tag*{\lipicsEnd} 
\]
\end{example}




\section{$\pi$-calculus: syntax and operational semantics}\label{sec:pi}
In this section we present the syntax and the early semantics of the $\pi$-calculus~\cite{pibook}.
We let $R, S$ and their primed versions range over processes and let $\Proc_{\pi}$ be the set of all processes,
and $\mathcal{N}_\pi$ the set of all channel names. Also, sometimes we will refer to elements
of $\mathcal{N}_\pi$ as names.

The syntax of $\pi$-calculus processes is reported below:
$$\begin{array}{lll}
R & ::= \Sigma_I \pi_i.R_i \ \mid \ R\Par S \ \mid \ \nu a.R \ \mid \ \rec X\,R  \ \mid \ X \\
\pi &::= a(x) \mid \Co{a}\tup{b} \mid \tau
\end{array}$$
As for CCSK, we consider $n$-ary guarded choice $\Sigma_I \pi_i.R_i$, but sometimes also consider prefixes $\pi_i.R_i$ in isolation to allow for more modular definitions.

Bound and free names of a process $R$ (written respectively $\bn(R)$ and $\fn(R)$)
are defined by saying that name $n$ is bound in $a(n).R$ and $\nu n. R$, other kinds
of occurrences defining free names. We indicate with $\tilde{x}$ a non-empty sequence of names $x_1,\cdots,x_n$,
and by abuse of notation we write $\nu \tilde{x}$, instead of $\nu x_1. \cdots. \nu x_n$.

The recursion operator $\rec X R$ binds occurrences of recursion variable $X$ in $R$.  Every occurrence of $X$ in $R$ must be \emph{guarded}, i.e.\ within a summation $\Sigma_I \pi_i.R_i$.  Within a process every recursion variable must be bound.

\begin{figure*}[ht!]
\begin{mathpar}
\inferrule*[lab={\scriptsize(\textsc{In})}]{}{a(x).R \pitran{av} R\subst{v}{x}} \and
\inferrule*[lab={\scriptsize(\textsc{Out})}]{}{\Co{a}\tup{v}.R \pitran{\Co{a}
\tup{v}} R} 
\and
\inferrule*[lab={\scriptsize(\textsc{Tau})}]{}{\tau.R \pitran{\tau} R} 
\and
\inferrule*[lab={\scriptsize(\textsc{Sum})}]{j \in I \\ R_j \pitran{\mu_j} R'_j }{ \Sigma_I R_i \pitran{\mu_j} R'_j} 
\and
\inferrule*[lab={\scriptsize(\textsc{Par-L})}]{R \pitran{\mu} R'  \and \bn(\mu) \cap \fn(S) = \emptyset}{R \Par S \pitran{\mu} R' \mid S} 
\and
\inferrule*[lab={\scriptsize(\textsc{Com-L})}]{R \pitran{\Co{a}\tup{v}} R' \and S\pitran{av} S'}{R \Par S \pitran{\tau} R' \mid S'} 
\and
\inferrule*[lab={\scriptsize(\textsc{Res})}]{R \pitran{\mu} R' \\ a\not \in \n(\mu) }{\nu a . R \pitran{\mu} \nu a. R'} 
\and
\inferrule*[lab={\scriptsize(\textsc{Open})}]{R \pitran{\Co{a}\tup{b}} R'  }{\nu b. R \pitran{\Co{a}(b)} R'} 
\and
\inferrule*[lab={\scriptsize(\textsc{Close-L})}]{R \pitran{ \Co{a}(b)} R'  \and S \pitran{ab} S' \and  b \not\in \fn(S)}{ R \Par S \pitran{\tau} \nu b.(R' \mid S')} 
\and

\inferrule*[lab={\scriptsize(\textsc{Rec})}]{R\subst{\rec X R}{X} \pitran\mu R'}
{\rec X R \pitran\mu R'}

\and
\inferrule*[lab={\scriptsize(\textsc{Str})}]{R \pistr R'  \and R' \pitran{\mu} S' \and  S'\pistr S}{ R \pitran{\mu} S} 
\end{mathpar}
\caption{$\pi$-calculus early labelled transition system. 
}
\label{fig:pi_sem}
\end{figure*}

The semantics of the $\pi$-calculus is given by the labelled 
transition system $(\Proc_{\pi}, \Act_{\pi},\pitran{})$, where
 ${\pitran{}} \subseteq \Proc_{\pi} \times \Act_{\pi} \times \Proc_{\pi}$
 is the smallest transition relation closed under the rules of Figure \ref{fig:pi_sem}. 
 The set of actions 
  $\Act_{\pi}$ is generated by the following grammar:
\[\begin{array}{lll}
\mu & ::= ab \ \mid \ \Co{a}\tup{b} \ \mid \ \Co{a}(b)   \ \mid \ \tau
\end{array}
\]
In the grammar, $ab$ represents an input with subject $a$ and object~$b$;
$ \Co{a}\tup{b}$ represents an output with subject $a$ and object $b$;
$\Co{a}(b)$ is a shorthand for  $\nu b \, \Co{a}\tup{b}$ and represents the output of a bound name;
and 
 $\tau$
denotes an internal step.
The bound and free names of an action $\mu$, written
respectively $\fn(\mu)$ and $\bn(\mu)$, are defined by saying that the name~$b$ is bound
in
$\Co{a}(b)$, while other occurrences of names are free. We also set 
$\n(\mu) = \fn(\mu) \cup \bn(\mu)$. 

We define $\pistr$ as the smallest congruence relation closed under the rules of Figure \ref{fig:pi_str},
where $R =_{\alpha} S$ indicates two processes equivalent modulo $\alpha$-conversion.
Relation $\wpitran{}$ is the transitive and reflexive closure of $\pitran{\tau}$, and
relation $\wpitran{\mu}$
 is defined as $\wpitran{}\pitran{\mu}\wpitran{}$. 

\begin{figure*}[t!]
\begin{align*}
& R =_{\alpha} S \implies R \pistr S \\
&R \Par \nil \pistr R &&  R \Par S \pistr S \Par R && R \Par (S \Par T) \pistr  (R \Par S) \Par T 
\\
 &R + \nil \pistr R &&  R + S \pistr S + R && R + (S + T) \pistr  (R + S) + T \\
&\nu a.\nil  \pistr \nil && \nu a.\nu b.R \pistr \nu b.\nu a. R && \nu a. (R \Par S) \pistr R \Par \nu a.S \,\, \text{ if } a \not \in \fn(P)
\end{align*}
\caption{$\pi$-calculus structural congruence.}
\label{fig:pi_str}
\end{figure*}

In our development, we will focus on a subcalculus of the $\pi$-calculus, called the internal $\pi$-calculus~\cite{Sangiorgi96a}. In the internal $\pi$-calculus, only \emph{new} names can be sent as outputs. In other words, output prefix $\Co{a}\tup{b}$ can only occur immediately inside a restriction name $b$, namely as a term $\nu b.\Co{a}\tup{b}$.  This implies that each synchronisation exchanges a different fresh name. As we will see, our encoding will use
the internal $\pi$-calculus of order two, where communicated names, beyond being bound, can then only be used for synchronisation (like CCS names), and not for sending further names (cf.~\cite[Definition 6.2]{Sangiorgi96a}). In this case we will drop the object from the prefix or the label.

\section{Separation result}\label{sec:separation}


We give conditions under which CCSK cannot be encoded into CCS or the $\pi$-calculus, adapting definitions and results from~\cite{PT20,Gor10}.
See App.~\ref{app:separation} for omitted proofs for this section.

We start by recalling the definitions of replacement freeness and basic encodings from~\cite{PT20}.  These are defined for general process calculi which have notions of reduction and barbs, which we shall instantiate for CCSK, CCS and the $\pi$-calculus.

\begin{definition}\label{def:calculus}
A \emph {(process) calculus} $\calc$ has a set of processes, ranged over by $P,\ldots$.  These are generated using process operators, giving rise to contexts as usual.
Processes use \emph{names} as communication channels and possibly for input parameters and output values.  Name-binding operators delimit the scope of names, and names may be either free or bound.
A process with no free names is \emph{closed}.
Calculus $\calc$ also has a notion of \emph{reduction} $P \red P'$, and \emph{barb}
$P \barb \alpha$, meaning that we can observe barb $\alpha$ at $P$.
Weak reduction $\wred$ is the reflexive and transitive closure of $\red$,
and $P \wbarb \alpha$ means that there is $P'$ such that $P \wred P' \barb \alpha$. 
\end{definition}
 
\begin{definition}
[Visibility \cite{PT20}]
A process $P$ is \emph{visible}, written $P \wbarb$, if $P \wbarb \alpha$ for some $\alpha$;
otherwise $P$ is \emph{invisible}.
\end{definition}

\begin{definition}
[Replacement freeness \cite{PT20}]\label{def:rep-free}
A calculus $\calc$ is \emph{strongly replacement free (strongly RF)} if for every single-hole context $C$, invisible process $I$ and process $P$ in $\calc$, we have
$C[I] \wbarb$ implies $C[P] \wbarb$.
A calculus $\calc$ is \emph{replacement free (RF)} if for every single-hole context $C$, \emph{closed} invisible process~$I$ and process $P$ in $\calc$, we have
$C[I] \wbarb$ implies $C[P] \wbarb$.
Furthermore, $\calc$ is \emph{weakly RF} if it is RF but not strongly RF.
\end{definition}

\begin{definition}
[Basic encoding {\cite[Def. 3.1]{PT20}}]\label{def:basic}
An encoding $\enc \cdot$ of $\calc_1$ into $\calc_2$ is \emph{basic} if:
\begin{enumerate}
\item 
$\enc \cdot$ is \emph{compositional}:
for every $k$-ary operator $\opr$ in $\calc_1$ there is a $k$-hole context $C_{\opr}$ in $\calc_2$ such that $\forall P_1, \ldots, P_k \in \calc_1$ we have $\enc{ \opr(P_1,\ldots, P_k)} = C_{\opr}[\,\enc{P_1},\ldots,\enc{P_k}\,]$.
\item
$\enc \cdot$ is \emph{interaction sensitive}: for all processes $P$ in $\calc_1$ we have $P \wbarb$ iff $\enc P \wbarb$.
\end{enumerate}
\end{definition}
If an encoding is compositional then for any $\calc_1$ context $C_1[\hol]$ there is a $\calc_2$ context $C_2[\hol]$ such that for all $\calc_1$ processes $P$ we have $\enc{C_1[P]} = C_2[\,\enc P \,]$~\cite[Lemma 3.1]{PT20}.

\begin{proposition}[{\cite[Thm.\ 3.1]{PT20}}]\label{prop:basic strongly RF}
There exists no basic encoding from a non-strongly RF calculus to a strongly RF one.
\qed
\end{proposition}
Proposition~\ref{prop:basic strongly RF} is used in~\cite{PT20} to show that the $\pi$-calculus extended with polyadic synchronisation and/or matching cannot be encoded into the basic $\pi$-calculus.

We now instantiate reduction and barbs for CCS, the $\pi$-calculus and CCSK.
%
\begin{definition}
[Reduction and barbs for CCS]
Let $P \red P'$ iff $P \tran {\tau} P'$.
Strong barb: $P \barb a$ if $P \tran a P'$ for some $P'$.
$P \barb \out a$ if  $P \tran {\out a} P'$ for some $P'$.
\end{definition}

%

\begin{definition}
[Reduction and barbs for the $\pi$-calculus]\label{def:pi barb}
Let $R \red R'$ iff $R \pitran {\tau} R'$. 
Strong barb: $R \barb a$ if $R$ can perform an input transition with subject $a$;
$R \barb \out a$ if $R$ can perform an output transition with subject $a$.
\end{definition}

In CCSK we allow both forward and backward barbs.
%

\begin{definition}
[Reduction and barbs for CCSK](cf.~\cite{LaneseP21})
Let $P \red P'$ iff either $P \ftran {\tau[m]} P'$ or $P \rtran {\tau[m]} P'$, for some key $m$.
Strong barb: $P \barb \mu$ if $P \ftran {\mu[m]} P'$ for $\mu \neq \tau$ and some $m$ and $P'$.
$P \barb \rev\mu$ if $P \rtran {\mu[m]} P'$ for $\mu \neq \tau$ and some $m$ and $P'$.
\end{definition}
The next lemma is instrumental in proving that CCS and the $\pi$-calculus are strongly RF.
\begin{restatable}{lemma}{CCSpiinvisiblebarb}\label{lem:CCS pi invisible barb}
In CCS or the $\pi$-calculus, let $C[\hol]$ be a context, let $I$ be invisible and let $P$ be any process, and let $\alpha$ be any barb (so that $\alpha = a$ or $\alpha = \out a$ for some name $a$).
If $C[I] \wbarb \alpha$ then $C[P] \wbarb \alpha$.
\end{restatable}
\begin{restatable}[{\cite{PT20}}]{proposition}{CCSpistronglyRF}\label{prop:CCS pi strongly RF}
CCS and the $\pi$-calculus are strongly RF.
\end{restatable}

%

The analogue of Lemma~\ref{lem:CCS pi invisible barb} does not hold for CCSK, as the following example shows.
\begin{example}\label{ex:CCSK barbs}
Let $C[\hol] = a[m].\hol$.  Let $I = \nil$ and $P = b[n]$.
Then $C[I] \wbarb \rev a$ but $C[P] \wbarb \rev b$ and not $C[P] \wbarb \rev a$. \lipicsEnd 
\end{example}
Nevertheless, we can show that CCSK is strongly RF using different methods.

\begin{restatable}{proposition}{CCSKstronglyRF}\label{prop:CCSK strongly RF}
CCSK is strongly RF. 
\end{restatable}
Since CCSK, CCS and the $\pi$-calculus are all strongly RF, we need to adapt the definitions of replacement freeness and basic encodings to obtain an impossibility result for encoding CCSK into the $\pi$-calculus.

We therefore modify Definition~\ref{def:rep-free} by considering success rather than visibility.  We suppose $\tick$ be a particular barb that can be used to record the success of a computation.
In CCSK, $\tick$ belongs to $\mathcal{N}$, the set of names, and is a forward barb;
similarly for CCS.
In the $\pi$-calculus, $\tick$ belongs to $\mathcal{N}_\pi$, the set of channel names, and is an input barb.
\begin{definition}
[Success replacement freeness]\label{def:success rep-free}
A calculus $\calc$ is \emph{strongly success replacement free (strongly SuRF)} if for every single-hole context $C$, invisible process $I$ and process $P$ in $\calc$, we have
$C[I] \wbarb \tick$ implies $C[P] \wbarb \tick$.
A calculus $\calc$ is \emph{success replacement free (SuRF)} if for every single-hole context~$C$, \emph{closed} invisible process $I$ and process $P$ in $\calc$, we have
$C[I] \wbarb \tick$ implies $C[P] \wbarb \tick$.
Furthermore, $\calc$ is \emph{weakly SuRF} if it is SuRF but not strongly SuRF.
\end{definition}

\begin{definition}
[{\cite[Property 5]{Gor10}}]\label{def:success}
An encoding $\enc \cdot$  of $\calc_1$ into $\calc_2$ is \emph{success sensitive} if for all processes $P$ in $\calc_1$ we have $P \wbarb \tick$ iff $\enc P \wbarb \tick$.
\end{definition}

We can obtain an analogue of Proposition~\ref{prop:basic strongly RF}:
\begin{restatable}{proposition}{basicsuccessSuRF}
\label{prop:basic success SuRF}
There exists no basic, success-sensitive encoding from a non-strongly SuRF calculus to a strongly SuRF one.
\end{restatable}
\begin{restatable}{proposition}{CCSpistronglySuRF}\label{prop:CCS pi strongly SuRF}
CCS and the $\pi$-calculus are strongly SuRF.
\end{restatable}

\begin{proposition}\label{prop:CCSK not SuRF}
CCSK is not SuRF.
\end{proposition}
\begin{proof}
We consider the context $C[\hol] = (a[m].\hol \Par \bar a[m].\nil) \Par a. \tick.\nil$, closed invisible process $I = \nil$ and process $P = b[n].\nil$.
Then $C[I] \tran {\rev\tau[m]} (a.\nil \Par \bar a.\nil) \Par a.\tick.\nil \tran {\tau[k]} (a.\nil \Par \bar a[k].\nil) \Par a[k].\tick.\nil \barb \tick$, so that $C[I] \wbarb \tick$.
However $C[P]$ cannot perform any reduction, so that $C[P] \wbarb \tick$ fails to hold.
\end{proof}
The proof of Proposition~\ref{prop:CCSK not SuRF} uses a limited fragment of CCSK: just $\nil$, prefixing and top-level parallel composition. Actually, if we allow backward barbs as $\tick$, then we could use just
$C[\hol]=\tick[m].\hol$
(cf.\ Example~\ref{ex:CCSK barbs}). Alternatively, with forward barbs and $\tau$ prefixes we could use
$C[\hol]=\tick.\nil+\tau[m].\hol$, featuring choice instead of parallel composition. This shows that the result is robust. Indeed, analogous results can be proved for many reversible calculi, obviously including calculi that can embed CCSK such as
revTPL~\cite{BocchiLMY24}, but also others such as reversible broadcast CCS~\cite{Mez18}. This shows that this separation result is more a feature of the reversibility mechanism (indeed, all examples above rely on backward actions) than of the details of the calculus.
As a consequence, also the main result below is robust.
\begin{restatable}{theorem}{basicsuccsensCCSKpi}\label{thm:basic succ sens CCSK pi}
There is no basic, success-sensitive encoding from CCSK to CCS or the $\pi$-calculus.
\end{restatable}
Considering the counterexample in the proof of Proposition~\ref{prop:CCSK not SuRF}, we can rephrase Theorem~\ref{thm:basic succ sens CCSK pi} to state that there is no basic, success-sensitive encoding from CCSK with only top-level parallel composition to CCS or the $\pi$-calculus.
We shall present an encoding from CCSK with top-level parallel composition to the $\pi$-calculus (see Figure~\ref{fig:enc_strong}).
It must therefore fail to satisfy
at least one of compositionality, interaction sensitivity and success sensitivity.  Indeed, it is not compositional but it does satisfy the other two properties.

\begin{remark}
Since the counterexample in the proof of Proposition~\ref{prop:CCSK not SuRF} does not use summation, we can revise the definition of compositionality and of basic encoding to make no requirement on summation, and show that Theorem~\ref{thm:basic succ sens CCSK pi} still holds with this weaker notion of basic encoding.
\end{remark}

\begin{figure}
\begin{center}
\begin{tabular}{|c|c|}
\hline
\textbf{strongly RF} & \textbf{not RF} \\
CCS, $\pi$, CCSK  & CPG \\
\hline
\end{tabular}
\qquad\qquad
\begin{tabular}{|c|c|}
\hline
\textbf{strongly SuRF}  & \textbf{not SuRF} \\
CCS, $\pi$ & CCSK \\
\hline
\end{tabular}
\end{center}
\caption{Classification of calculi according to RF and SuRF.}\label{fig:RF SuRF}
\end{figure}
See Figure~\ref{fig:RF SuRF} for a summary of the classification of calculi according to RF and SuRF.  CPG is CCS with Priority Guards~\cite{Phi08}, which is shown to be not RF in~\cite[Prop.\ 6.2]{PT20}.

\section{Encoding CCSK into the $\pi$-calculus}\label{sec:encoding}

The previous section shows that there is no basic, success-sensitive encoding of CCSK into CCS 
or the $\pi$-calculus. 
Now we present an encoding of a subset of CCSK (with only top-level parallel), into the $\pi$-calculus.
%
We show that  only a subset of the $\pi$-calculus is needed, 
namely \emph{internal} $\pi$~\cite{Sangiorgi96a}.
We then prove that this encoding is correct up 
to strong bisimilarity.  
See App.~\ref{sec:app} for omitted proofs for this section.

CCSK processes use 
both names and keys.
 Both will be encoded as $\pi$-calculus names. We will call \emph{channel names} the encoding of the former and \emph{key names} the encoding of the latter.
We assume the two sets of $\pi$ names to be disjoint.
In CCSK, keys are created in forward computation and consumed to manage backward computation. 
Hence, in the $\pi$-calculus, images of CCSK forward computations will create new key names (via the $\nu$ operator) and images of CCSK backward computations will use these key names to go back to past states.

The operational semantics of CCSK guarantees that if a key occurs in a reachable process then it has one or two occurrences; 
see~\cite[Prop.\ 3.5]{LaneseP21}.
Keys are either free or bound, depending 
on which prefixes they are attached to (Definition~\ref{def:fbk}). Free and bound keys will be modelled by free and bound key names, respectively, in the $\pi$-calculus.

To track which key names are used to model which keys,
the encoding is parametric on a bijection $\phi$, which records the
correspondence between free 
keys and free 
key names.
In CCSK each key is always bound to a single name  (either as an input action $a$, or an output action $\Co{a}$, or both);
hence $\phi$ can be written as a set of items of the form $\phiel{a}{k}{x}$,
where $a, k$ are a CCSK name and a key, respectively, 
and $x$ is a $\pi$-calculus key name. 
We assume the usual operations on bijections: 
\emph{extension} $\phi[\phiel{a}{k}{x}]$, meaning that $\phi$ is extended with $\phiel{a}{k}{x}$ 
(assuming that $(a,k)$ is not in the domain of $\phi$);
and \emph{restriction} $\phi\setminus(a,k)$, meaning that $(a,k)$ is removed from the domain
of $\phi$.
The correspondence only tracks free keys, since bound keys correspond to bound names and hence are $\alpha$-convertible.

\begin{definition}[CCSK to $\pi$ encoding]\label{def:encoding}
The encoding function $\enc{\cdot}: \Proc\rightarrow \Proc_{\pi}$ is defined by
$\enc{P} = \nu K.\enc{P,\nil,\phi}$
where
$K$ is the set of key names for bound keys in $P$, 
and
$\enc{P,\nil,\phi}$ is defined in Figure~\ref{fig:enc_strong},
where $\phi$ is omitted since it is fixed.
\end{definition}
Note that if we encode a standard process then the set of bound keys is empty.
The second parameter of the encoding in Figure~\ref{fig:enc_strong} is a $\pi$-calculus process (ranged over by $R$), 
and is computed
by the encoding itself. 
It is used to build the previous state of the $\pi$-calculus process which is the encoding of a CCSK process $P$.
We will often refer to this parameter as a \emph{backtrack process}.  
%
\begin{figure}[t]
\begin{align*}
  \enc{P \Par Q, \nil} =\;& \enc{P,\nil} \Par \enc{Q,\nil}\\
  \enc{(\nu a)P , R} =\;& \nu a.\enc{P,R}\\
  \enc{\Sigma_i \mu_i.P_i,R}  =\;& \rec X (R + \Sigma_i \encpref{\mu_i.P_i,X}) \qquad \qquad \textrm{ $X$  fresh}\\
  \enc{\Co{a}[k].P + \Sigma_i \mu_i.P_i, R}  =\;& \enc{P,\Co{x_k}.\enc{\Co{a}.\tostd{P}+\Sigma_i \mu_i.P_i,R}} 
\quad \textrm{ with $\phiel{a}{k}{x_k} \in \phi$}\\
  \enc{a[k].P + \Sigma_i \mu_i.P_i, R}  =\;& \enc{P,x_k.\enc{a.\tostd{P}+\Sigma_i \mu_i.P_i,R}} 
\quad \textrm{ with $\phiel{a}{k}{x_k} \in \phi$}\\  
    \enc{\tau[k].P + \Sigma_i \mu_i.P_i, R}  =\;& \enc{P,\tau.\enc{\tau.\tostd{P}+\Sigma_i \mu_i.P_i,R}} 
  \end{align*}
 \[ \encpref{\Co{a}.P,X} =\; \Co{a}(y).\enc{P,\Co{y}.X} \qquad
  \encpref{a.P,X} =\; a(y).\enc{P,y.X} \qquad
 \encpref{\tau.P,X} =\; \tau.\enc{P,\tau.X}\]
\caption{Encoding of CCSK into the internal $\pi$-calculus.}
\label{fig:enc_strong}
\end{figure}

We now explain the encoding rules in Figure~\ref{fig:enc_strong}.
At the top-level, the only place in which parallel
composition is allowed,
the backtrack process $R$ is $\nil$.
The parallel and restriction
operators are encoded homomorphically. 
A process $\Sigma_i \mu_i.P_i$ represents an $n$-ary choice among~$n$ prefixed processes. 
Since CCSK is a reversible process calculus, we can get back
to the original process from the executed branch. 
%
We enable backward execution via the recursion variable $X$, which gives access to the state before the choice is taken. 
Variable $X$ is fresh and 
is passed to all the branches of the encoding of the choice. In this way, every branch has enough information
to get back to the previous state $X$.
%
Hence, the encoding of the $n$-ary choice is rendered as an
$(n+1)$-ary choice in which the first branch is the backtrack process~$R$. 
This idea is illustrated  in Example~\ref{lbl:enc_cho} below.

A process 
$\Co{a}[k].P + \Sigma_i \mu_i.P_i$ represents a choice process, where the first branch $\Co{a}[k].P$
is being executed (and $\Sigma_i \mu_i.P_i$ is unused).
Hence, the encoding proceeds to encode the prefix continuation $P$.
Meanwhile, it builds the parameter $R$ in such a way that once 
the action $\Co{a}[k]$ is reversed via the prefix $\Co{x_k}$, the process evolves to the encoding of 
the entire choice process $a.\tostd{P}+\Sigma_i \mu_i.P_i$. 
%
Since the process $P$ may have executed actions with keys, 
we use the erasing function $\tostd{\cdot}$ to remove them, obtaining a standard version of $P$. 

\Comment{
We define the encoding as $\enc{P,\phi}=\nu X_K \enc{P,\irek{R}, 
\phi}$ where
$X_K$ is the set of \il{bound} key names \il{in $\enc{P,\irek{R},\phi}$}.

The second part, described in Figure~\ref{fig:enc_strong}, is an encoding with three parameters: the CCSK process
$P$ to encode, the backtrack process \irek{$R$} of \irek{$P$}, if
any, $\nil$ otherwise, and the correspondence $\phi$ mentioned
above. Since the correspondence $\phi$ is fixed, we do not write it
explicitly, but only mention it when relevant.
}
To simplify the definition of the encoding, and make it more modular, we define it on single prefixes (last line of the encoding) and then use this definition to specify the encoding of choice.
The encoding of a prefix
$\Co{a}.P$ is rendered as an output prefix which sends on the channel name $a$ a new key name $y$,
and continues as the encoding of its continuation. The freshly created
key name is used to synchronise the undo of the output with the corresponding input. This is why the backtrack process $X$ is prefixed with an action on key name $y$.
The encodings of a keyed input prefix and a $\tau$ prefix are similar. The only difference
is that the encoding of the input receives the fresh key name from the output one,
while we do not need a key name for a $\tau$ prefix.


%



The image of the above encoding is a subset of internal $\pi$
of order two $\pi\mathbf{I}^2$, since communicated names are used for synchronisation only. According to~\cite{Sangiorgi96a}, CCS and  $\pi\mathbf{I}^2$ have distinct expressive power.
This result is not however explicitly lifted in~\cite{Sangiorgi96a} to an impossibility result of encoding internal $\pi$ into CCS.
See Section~\ref{sec:conc} for further discussion.

\begin{example}\label{lbl:enc_cho}
Consider the CCSK process $a.\Co{b}.\nil$. Since $a.\Co{b}.\nil$ is standard, $\phi$ is $\emptyset$. 
The encoding $\enc{a.\Co{b}.\nil}$ is $\nu \emptyset. \enc{a.\Co{b}.\nil,\nil,\phi}$. 
Simplifying and omitting $\phi$, we calculate $\enc{a.\Co{b}.\nil,\nil}$ as
     \begin{align*}
    &
    \rec{X_a} (\nil + a(x_a).\enc{\Co{b}.\nil,x_a.X_a} )
     = \rec{X_a} a(x_a).\rec{Y_{\Co{b}}} (x_a.X_a + \Co{b}(y_b).\enc{\nil,\Co{y_b}.Y_{\Co{b}}})\\
      &  \quad                     = \rec{X_a} a(x_a).\rec{Y_{\Co{b}}} (x_a.X_a + \Co{b}(y_b).(\Co{y_b}.Y_{\Co{b}} + \nil) \\
        &  \quad                   
        = \rec{X_a} a(x_a).\rec{Y_{\Co{b}}} (x_a.X_a + \Co{b}(y_b).\Co{y_b}.Y_{\Co{b}})
  \end{align*}
\Comment{Let us now consider the CCSK process $a.\Co{b}.\nil$. Since the process is an initial one, the map $\phi$
  is the $\emptyset$. The encoding is then
     \begin{align*}
    &\enc{a.\Co{b}.\nil,\emptyset} =     \enc{a.\Co{b}.\nil,\nil,\emptyset} =  \enc{a.\Co{b}.\nil,\nil}\\
    &\quad			 = \rec{X_a} (\nil + a(x_a).\enc{\Co{b}.\nil,x_a.X_a} )\\
     & \quad                       = \rec{X_a} ( \nil + a(x_a).\rec{Y_{\Co{b}}} (x_a.X_a + \Co{b}(y_b).\enc{\nil,\Co{y_b}.Y_{\Co{b}}}))\\
      &  \quad                     = \rec{X_a} ( \nil + a(x_a).\rec{Y_{\Co{b}}} (x_a.X_a + \Co{b}(y_b).(\Co{y_b}.Y_{\Co{b}} + \nil))\\
        &  \quad                   = \rec{X_a} a(x_a).\rec{Y_{\Co{b}}} (x_a.X_a + \Co{b}(y_b).\Co{y_b}.Y_{\Co{b}})
  \end{align*}
}
As highlighted by the choice of the names of the recursion variables (which is immaterial, since they are bound), undoing $a$ by synchronising on channel $x_a$ leads back to $X_a$, and analogously undoing $\Co{b}$ by synchronising on channel $y_b$ leads back to $Y_{\Co{b}}$. \lipicsEnd 
\end{example}

\begin{example}
Consider process 
$P = a[k].(b[h].\nil + c.\nil)$.
Letting $\phi = \{\phiel{a}{k}{x_k}, \phiel{b}{h}{x_h}\}$, and
noting that $K$ is empty since both keys are free, $\enc{P,\nil,\phi}$ is calculated as
\begin{align*}
   & \enc{(b[h].\nil + c.\nil), x_k.\enc{a.\tostd{b[h].\nil + c.\nil}, \nil}}   
     = \enc{(b[h].\nil + c.\nil), x_k.\enc{a.(b.\nil + c.\nil), \nil}}   \\
     &\quad= \enc{\nil, x_h.\enc{b.\nil + c.\nil, x_k.\enc{a.(b.\nil + c.\nil), \nil}}} 
     =\nil + x_h.\enc{b.\nil + c.\nil, x_k.\enc{a.(b.\nil + c.\nil), \nil}} \\
     &\quad=x_h.\textcolor{blue}{\enc{b.\nil + c.\nil, x_k.\textcolor{red}{\enc{a.(b.\nil + c.\nil), \nil}}}}
\end{align*}
where (for the sake of readability, we use colours to highlight the substituted subterms):
\begin{align*}
&\enc{a.(b.\nil + c.\nil), \nil} = \textcolor{red}{\rec{X_k}\encpref{a.(b.\nil + c.\nil), X_k}}\\
&\enc{b.\nil + c.\nil, x_k.\enc{a.(b.\nil + c.\nil), \nil}} = \textcolor{blue}{\rec{Y_h}(x_k.{\color{red}{\enc{a.(b.\nil + c.\nil),\nil}}}} 
 + \textcolor{blue}{\encpref{b.\nil, Y_h} +   
\encpref{c.\nil, Y_h}  )}.
\end{align*}
Assuming $R=x_h.\textcolor{blue}{\rec{Y_h}(x_k.{\color{red}{\rec{X_k}\encpref{a.(b.\nil + c.\nil), X_k}}} + \encpref{b.\nil, Y_h} +   
\encpref{c.\nil, Y_h}  )}$, we have $\enc{P,\nil,\phi}= \nu \emptyset. R$. We shall omit the outermost 
restriction in the rest of the paper.

If we look at the process $P$, the only move it can make is  $P\rtran{b[h]} a[k].(b.\nil + c.\nil)$, 
to a process which can still either undo $a[k]$ or move forward by doing $b$ or $c$. 
This behaviour is perfectly matched by $R$ since the only action that $R$ can perform is $x_h$ 
(mimicking the undoing of $b[h]$):
\begin{align*}
R \pitran{x_h} 
\quad\textcolor{blue}{\rec{Y_h}(x_k.{\color{red}{\rec{X_k}\encpref{a.(b.\nil + c.\nil), X_k}}} + \encpref{b.\nil, Y_h} +   
\encpref{c.\nil, Y_h}  )} \tag*{\lipicsEnd}
\end{align*}
\Comment{The target process can either do $x_k$ (mimicking the undoing of $a[k]$) or do either $b$ or $c$. 
Let us note that
the process variables $Y_h$ and $X_k$ are used by the encoding to point to the \emph{previous} states of $R$, namely after $a$ is
performed and before $a$ is performed, respectively.\finex}
\end{example}

\begin{example}
Let us consider the CCSK processes $P_1 = a.\Co{b}.\nil$ and  $P_2 = b.\nil$.
According to CCSK semantics we can derive the following computation:
\begin{align*}
P_1 \Par P_2 &\ftran{a[k]} a[k].\Co{b}.\nil \Par b.\nil \ftran{\tau[w]}a[k].\Co{b}[w].\nil \Par b[w].\nil \rtran{\tau[w]} 
a[k].\Co{b}.\nil \Par b.\nil
\end{align*}
The encoding of $P_1 \Par P_2$ is $\enc{P_1\Par P_2,\nil}$, which we calculate as
\begin{align*}
&\enc{P_1, \nil} \Par \enc{P_2,\nil} =
 \rec{X_a} a(x_a).\rec{Y_{\Co{b}}} (x_a.X_a + \Co{b}(y_b).\Co{y_b}.Y_{\Co{b}})  \Par 
 \rec{Y_b} b(y_b).y_b.Y_b
 \end{align*}
 The encoded process reproduces the transitions of $P_1 \Par P_2$ as follows. 
 Let us note that for the sake of readability we do not substitute the recursive process variables. 
 \begin{align}
&  \enc{P_1\Par P_2,\nil}\pitran{ax_a} \nonumber\\
&  \rec{Y_{\Co{b}}} (x_a.X_a + \Co{b}(y_b).\Co{y_b}.Y_{\Co{b}})  \Par 
 \rec{Y_b} b(y_b).y_b.Y_b \pitran{\tau}\label{proc_b}\\
 &\nu y_b.(\Co{y_b}.Y_{\Co{b}} \Par y_b.Y_{b}) \pitran{\tau} \label{proc_tau}\\
 & \nu y_b. (  \rec{Y_{\Co{b}}} (x_a.X_a + \Co{b}(y_b).\Co{y_b}.Y_{\Co{b}})  \Par 
 \rec{Y_b} b(y_b).y_b.Y_b ) \pistr \label{proc_back}\\
 & \rec{Y_{\Co{b}}} (x_a.X_a + \Co{b}(y_b).\Co{y_b}.Y_{\Co{b}})  \Par 
 \rec{Y_b} b(y_b).y_b.Y_b  \label{proc_eq}
\end{align}
The two processes in (\ref{proc_b}) can synchronise on $b$ and reach
the process (\ref{proc_tau}), from which the two processes can \emph{undo} the synchronisation
on $b$ by synchronising on the restricted key name $y_b$. This synchronisation will lead
to the process (\ref{proc_back}), and by garbage collecting the restricted key name $y_b$ via $\pistr$ we get to process
(\ref{proc_eq}), which is equal to process (\ref{proc_b}). \lipicsEnd 
\end{example}

\subsection{Correctness}
To prove the correctness of our encoding w.r.t.~strong bisimulation, we have to recast the classical definition of strong bisimilarity~\cite{pibook} to work on two different semantics and calculi.
\begin{definition}[CCSK-$\pi$ strong bisimilarity]\label{def:bisim}
  A relation $\rcal$ between CCSK processes and $\pi$ processes
 parametrised with a bijection $\phi$ from CCSK names and keys to $\pi$ key names is a
  \emph{CCSK-$\pi$ bisimulation} iff $(P,R)_{\phi} \in \rcal$ implies the following: 
  \begin{itemize}
  \item if $P \ftran{\Co{a}[k]} P'$ then $R \pitran{\Co{a}(x)} R'$ with $(P',R')_{\phi[\phiel{a}{k}{x}]} \in \rcal$; 
  \item if $P \ftran{a[k]} P'$ then $R \pitran{ax} R'$ with $(P',R')_{\phi[\phiel{a}{k}{x}]}\in \rcal$;
  \item if $P \ftran{\tau[k]} P'$ then $R \pitran{\tau} R'$ with $(P',R')_{\phi}\in \rcal$;  
  \item if $P \rtran{\Co{a}[k]} P'$ then $R \pitran{\Co{x}} R'$ with $(P',R')_{\phi{\setminus (a,k)}}\in \rcal$
   where $\phi(a,k)=x$;
  \item if $P \rtran{a[k]} P'$ then $R \pitran{x} R'$ with $(P',R')_{\phi{\setminus (a,k)}}\in \rcal$
   where $\phi(a,k)=x$;
  \item if $P \rtran{\tau[k]} P'$ then $R \pitran{\tau} R'$ with $(P',R')_{\phi}\in \rcal$.  
  \end{itemize}
  \Comment{\begin{itemize}
  \item if $P \ftran{\Co{a}[k]} P'$ then $R \pitran{\Co{a}(x)} R'$ with $(P',R')_{\phi[(a,k) \mapsto x]} \in \rcal$ 
  \item if $P \ftran{a[k]} P'$ then $R \pitran{ax} R'$ with $(P',R')_{\phi[(a,k) \mapsto x]} \in \rcal$
  \item if $P \ftran{\tau[k]} P'$ then $R \pitran{\tau} R'$ with $(P',R')_{\phi}\in \rcal$;  
  \item if $P \rtran{\Co{a}[k]} P'$ then $R \pitran{\Co{x}} R'$ with $(P',R')_{\phi_\irek{{\setminus (a,k)}}}\in \rcal$ where $\phi(a,k)=x$;
  \item if $P \rtran{a[k]} P'$ then $R \pitran{x} R'$ with $(P',R')_{\phi_{\setminus (a,k)}}\in \rcal$ where $\phi(a,k)=x$;
  \item if $P \rtran{\tau[k]} P'$ then $R \pitran{\tau} R'$ with $(P',R')_{\phi}\in \rcal$;  
  \end{itemize}}
  and vice versa.
  We denote with $\bisim$, dubbed \emph{CCSK-$\pi$ strong bisimilarity}, the largest CCSK-$\pi$ strong bisimulation.
\end{definition}

As mentioned before, $\phi$ only tracks ``free names'', which in the case
of CCSK means keys with just one occurrence inside the term, attached
to a non-$\tau$ prefix. This is needed since such occurrences
correspond to synchronisations with the context, which contains the
other occurrence of the name. Hence, $\phi$ is needed to remember how
we translated keys whose corresponding name is known by the context,
to ensure that the CCSK process and its $\pi$-calculus encoding remain
aligned on which name represents which key. In the case of bound keys,
the
corresponding key name is bound in the $\pi$-calculus and can be
$\alpha$-converted. There is no need to track the correspondence with
such a name, since the name will never occur in the context.

We now define a function $\varphi$ that, given a CCSK process $P$, builds the concrete relation used in its encoding. Function $\varphi$ is actually a bijection, thanks to \cite[Prop.\ 3.5]{LaneseP21}.
\begin{definition}
Let $\varphi$ be a bijection calculated inductively on the CCSK syntax as follows:
\begin{align*}
&\varphi(P) = \emptyset \text{ if } \std{P}  
&& \varphi(\tau\key{k}.P) =\varphi(P) \\
&\varphi(a\key{k}.P) = \varphi(\Co{a}\key{k}.P) = \{\phiel{a}{k}{x_k}\} \cup \varphi(P)
&& \varphi(P+Q) =\varphi(P) \cup \varphi(Q) \\
&\varphi(P \Par Q) =\varphi(P) \cup \varphi(Q) \setminus (\varphi(P) \cap \varphi(Q)) 
&& \varphi(\nu a.P) =\varphi( P) 
\end{align*}
\end{definition}

\Comment{
We now present a few lemmas needed to prove the correctness of the
encoding.

\begin{lemma}\label{lemma:ancestor}
  If $\enc{P,\nil} \pitran{\mu} \enc{P',\nil}$ then
  $\enc{P,R} \pitran{\mu} \enc{P',R}$.
\end{lemma}
\begin{proof}
  By structural induction on $P$, with a case analysis on the last
  applied rule.
\end{proof}

\begin{lemma}\label{lemma:str}
Let $P$, $P'$ be CCSK processes and $R$, $R'$ be $\pi$-processes. If $P \str P' \bisim R' \pistr R$ then $P \bisim R$.
\end{lemma}
\begin{proof}
  Follows since both $\str$ and $\pistr$ preserve and reflect transitions.
\end{proof}
}

\begin{restatable}{theorem}{strongencbis}\label{th:strong_enc}
Let $P$ be a CCSK process with parallel composition only at the top level.
Then $(P,\enc{P,\nil,\varphi(P)}) \in {\bisim}$.
\end{restatable}
\Comment{
\begin{proof}
  It is sufficient to show that the relation
  \begin{align*}
    \rcal=&\Big\{\pair{P}{\enc{P,\nil,\varphi(P)}} \, | \, P \textrm{ is a CCSK process} \big\} 
  \end{align*}
  is a
  \emph{CCSK-$\pi$ bisimulation}. The proof
is by structural induction on $P$, and for each $P$ by coinduction. The full proof can be found in Appendix~\ref{sec:app}.
  \end{proof}
}
  
Strong bisimilarity, which we use as the correctness criterion in Theorem~\ref{th:strong_enc}, 
is stronger than both interaction sensitivity, used in the definition of
basic encoding (cf.~Definition~\ref{def:basic}), and success sensitivity (cf.~Definition~\ref{def:success}). The reason why the
result above is not in contrast with the separation result (Theorem~\ref{thm:basic succ sens CCSK pi}) 
in Section~\ref{sec:separation} is that the encoding is not compositional.
%
Indeed, the encoding of $P$ in a process such as $a[k_1].P$ depends 
on its context $a[k_1].\bullet$. More precisely, by looking at the encoding in Figure~\ref{fig:enc_strong}, 
one can see that $P$ is encoded in the case of a choice as a recursive process containing itself a choice, whose first alternative is the backtrack process $R$, 
needed to implement the backward move undoing $a[k_1]$. Reversing to $a.P$ requires knowing $a$ (and how it has been encoded), and cannot be deduced by looking at $P$ alone. This is indeed why $R$ is needed
as a parameter of the encoding.
%
%
When $R$ is not $\nil$, the encoding is not interaction
sensitive either, since the encoding of $\nil$ is $R$. This does not spoil our bisimilarity result, since at the top level $R$ is always $\nil$.

The encoding presented in Figure~\ref{fig:enc_strong} only considers
 top-level parallel composition (cf.~Definition~\ref{def:CCSK top-level}).
 %
%
If instead, in $a[k_1].P$, process $P$ is composed of multiple parallel processes, they need
to coordinate to decide whether the undo of $a$ is enabled.  Such
an interaction is at least $n$-ary, where $n$ is the number of
parallel components, but this is done in a single CCSK step, thanks to predicate $\std{Q}$ in rule {\scriptsize(\textsc{Act1}$^{\bullet}$)} of Figure~\ref{fig:ccsk_sem_full}. 
Hence, no encoding correct 
w.r.t. strong bisimilarity
preserving the degree of parallelism can exist.

In the literature (e.g., in~\cite{Palamidessi03}), the concept of preserving the degree of parallelism is frequently formalised by stating that the encoding is homomorphic w.r.t.\ parallel composition, namely that $\enc{P_1 \Par P_2}=\enc{P_1} \Par \enc{P_2}$. Indeed, the encoding above actually satisfies this property.

However, such a property can be useful when reasoning about lower-level parallel composition, especially when combined with a form of compositionality for other operators, prefix in particular.
One could require that for each CCSK prefix $\rho$ there exists a $\pi$-calculus context $C_\rho$ such that the encoding of $\rho.P$ is $C_\rho[\enc{P}]$. However, this would lead to a compositional encoding, but as we have shown in Section~\ref{sec:separation}  this is not possible (paired with interaction sensitivity; see Theorem~\ref{thm:basic succ sens CCSK pi}).

Hence we require a weaker form of compositionality, allowing the context to influence the encoding, but preserving the number of parallel components, and requiring parallel components to be encoded independently.


\begin{definition}\label{def:nary par pres enc}
An encoding is \emph{parallel preserving} if for every context $C[\hol]$ and every $n \geq 2$ there is a context $G[\hol]$, such that for all $P_1,\ldots,P_n$ with no top-level parallel composition there are $R_1,\ldots,R_n$ such that $\enc{C[P_1 \Par \cdots \Par P_n]} = G[R_1 \Par \cdots \Par R_n]$.  Also, $R_i$ only depends on $C,n,i$ and $P_i$ ($i = 1,\ldots,n$).
\end{definition}

The last condition in the definition means that the encodings $R_i$ of different parallel components are independent.

\begin{restatable}{proposition}{encisparallelpreserving}\label{prop:encpar}
The encoding in Figure~\ref{fig:enc_strong} is parallel preserving.
\end{restatable}

The encoding presented in Figure~\ref{fig:enc_strong} is
parallel preserving and 
preserves strong bisimilarity; however, it is defined on the fragment of CCSK with
only top-level  parallel composition.



\begin{restatable}{theorem}{noParPresCCSKtopipressb}\label{thm:par pres strong bis}
No parallel-preserving encoding of CCSK into the $\pi$-calculus preserves strong bisimilarity.
\end{restatable}



\section{Encoding lower-level parallel composition}\label{sec:weak_encoding}

The encoding in \Cref{sec:encoding} is only defined for CCSK processes with top-level parallel composition. We
have shown that no parallel-preserving encoding of CCSK (with parallel composition at lower levels as well as at top level) 
into the $\pi$-calculus can be correct with respect to strong bisimilarity. However, in this section, we extend our encoding to cover the whole of CCSK, and show 
that the resulting encoding is parallel preserving. The price to pay is that we can only state its correctness under a weaker notion,
namely \emph{weak mutual simulation}~\cite{Milner71}.
 See App.~\ref{app:weak} for omitted proofs for this section.
 
The main challenge in extending the encoding is in the implementation of backward moves, since a backward move is enabled only if its continuation is standard (cf.~the hypothesis of rule {\scriptsize(\textsc{Act1}$^{\bullet}$)} in Figure~\ref{fig:ccsk_sem_full}). 
In the encoding in \Cref{fig:enc_strong}, a standard process triggered the backward move via a process $R$. Now, however, all the processes in an arbitrary parallel composition need to agree on allowing a backward move: they all need to be standard.
To this end, we define the $\tree(N,R)$ function, which
given a set $N$ of names and a process $R$ generates a sequential
$\pi$ process that implements a ``tree'' of all possible sequences of
inputs on names in $N$ leading to $R$.
This function builds
a process which waits for ``rollback signals'', as outputs on names in $N$, from all the parallel subprocesses.
When they are all received, meaning that all subprocesses are standard,
function $\tree(N,R)$ triggers rollback by executing $R$. Hence, it needs to account for all the possible interleavings of ``rollback signals''.
\begin{example}
Consider $N=\{a,b,c\}$, and some $R$.  We want to construct a $\pi$ process that requires synchronisations on names $a, b,c$, 
in any order before reaching $R$. This process is  
$a.(b.c.R + c.b.R) + b.(a.c.R + c.a.R) + c.(a.b.R + b.a.R)$; it has a tree-like structure. 
\lipicsEnd 
\end{example}
\begin{definition}[Extended encoding]\label{def:tree}
The function $\tree : 2^\mathcal{N_\pi}\times \Proc_{\pi} \rightarrow \Proc_{\pi}$ is given by: 
\[
\tree(N,R) = \left\{\begin{array}{ll}R & \text{if } |N| = 0; \\
 \Sigma_{a\in N} a.(\tree(N\setminus \{a\},R)) & \text{otherwise}. \\ 
\end{array}\right.
\]
%
%
\noindent We extend the encoding in Figure~\ref{fig:enc_strong} with the following clause:
\begin{align*}
&\enc{\prod_I P_i,R} \,=\, \nu \tilde{x}.\big(\prod_I \enc{P_i,\Co{x_i}} \Par \tree(\tilde{x},R)\big)
\end{align*}
where
$\tilde{x}=x_1,\cdots,x_n$ and $|I| = |\tilde{x}|$, 
and each $P_i$  does not contain any top-level parallel composition operator.
\end{definition}

Process $\tree(\tilde{x},P)$
waits for rollback signals $\Co{x_i}$ from every parallel component and, when all of them have been received, it triggers the rollback process $R$. 
If $R$ is $\nil$ the system terminates, while in CCSK this means that the process is standard and so backward moves are disabled. To avoid this issue, top-level parallel composition is encoded 
with the rule in Figure~\ref{fig:enc_strong}.


\begin{remark}
The encoding above allows for rollback messages $\Co{x_i}$ to be
received in any order. Picking a specific order would not change the
correctness result, but will make the encoding of structural
congruent processes different, making the proof more complex. This
would be however better in a practical setting, since it reduces the
size of the resulting process (from a factorial of the number of
parallel components to linear).
\end{remark}

\Comment{\begin{remark}\label{rem:binenc}
The encoding, as extended in Definition~\ref{def:tree}, translates all the components of a parallel
composition at once. One could also translate individual parallel operators, thus
obtaining a binary tree of processes receiving the synchronisations. Again,
the correctness result would be the same, but structural congruent
processes would have different translations. Also, a higher number of
synchronisations would be needed, making the encoding less efficient.
\cm{we could remove this remark}
\end{remark}
}

\begin{example}
Let us consider the encoding of the CCSK process $a.(b.\nil\Par c.\nil)$.
  \begin{align*}
&    \enc{a.(b.\nil \Par c.\nil),\nil} = \rec{X_a}  ( \nil + \encpref{a.(b.\nil \Par c.\nil), X_a})\\
&= \rec {X_a} \big(\nil + a(y_a).(\enc{b.\nil\Par c.\nil,y_a.X_a}) \big)\\
&  =\rec {X_a} \big(\nil + a(y_a).(\nu x_1,x_2. (\enc{b.\nil,\Co{x_1}}    \Par\enc{c.\nil,\Co{x_2}}  \Par  \tree({\{x_1,x_2\}, y_a.X_a})\big)  \\
&  =\rec {X_a} \big(\nil + a(y_a).(\nu x_1,x_2. (\enc{b.\nil,\Co{x_1}}    \Par\enc{c.\nil,\Co{x_2}}  \Par x_1.x_2.y_a.X_a + x_2.x_1.y_a.X_a)\big)  \\
&=  \rec {X_a} \big( a(y_a).(\nu x_1,x_2. (\rec{X_b}(\Co{x_1} + \encpref{b.\nil,X_b})  \Par \\
 &  \qquad\qquad \rec{X_c}(\Co{x_2} + \encpref{c.\nil,X_c})  \Par x_1.x_2.y_a.X_a + x_2.x_1.y_a.X_a)\big)\\
        & = \rec {X_a} \big( a(y_a).(\nu x_1,x_2. (\rec{X_b}(\Co{x_1} + b(y_1).y_1.X_b)  \Par  \\
    & \qquad\qquad  \rec{X_c}(\Co{x_2} + c(y_2).y_2.X_c)  \Par x_1.x_2.y_a.X_a + x_2.x_1.y_a.X_a)\big) \tag*{\lipicsEnd}
    \end{align*} 
\end{example}

\begin{example}
Consider the encoding of $P = a[h].(b[k].\nil\Par c[l].\nil)$ with a bijection $\phi = \{\phiel{a}{h}{x_h}, \phiel{b}{k}{x_k}, \phiel{c}{l}{x_l}\}$.
We have
\begin{eqnarray*}
    \enc{P,\phi} &=&   \enc{(b[k].\nil \Par c[l].\nil), x_h.\enc{a.\tostd{b[k].\nil \Par c[l].\nil}, \nil}} \\
       &=&
      \enc{(b[k].\nil \Par c[l].\nil), x_h.\enc{a.(b.\nil \Par c.\nil), \nil}}   \\
     &=& \nu x_1,x_2.\big( \enc{b[k].\nil, \overline{x_1}} \Par \enc{c[l].\nil, \overline{x_2}} \Par
      \tree( \{x_1,x_2\},x_h.\enc{a.(b.\nil \Par c.\nil)})) \big)
\end{eqnarray*}
We obtain $\tree(\{x_1,x_2\},x_h.\enc{a.(b.\nil \Par c.\nil)})= x_1.x_2.x_h.\enc{a.(b.\nil \Par c.\nil)} + x_2.x_1.x_h.\enc{a.(b.\nil \Par c.\nil)}$, and abbreviating the tree process as $\tree(\cdots)$ below, we continue working out $\enc{P,\phi}$: 
\begin{eqnarray*}     
          &=& \nu x_1,x_2.\big( \enc{\nil, x_k.\enc{b.\nil, \overline{x_1}}} \Par  \enc{\nil, x_l.\enc{c.\nil,  \overline{x_2}}}  \Par \tree(\cdots)\big)\\
        &=& \nu x_1,x_2.\big(  x_k.\enc{b.\nil, \overline{x_1}} \Par x_l.\enc{c.\nil,  \overline{x_2}}  \Par\tree(\cdots) \big)\\
        &=& \nu x_1,x_2.\big(  x_k.\rec{X_b}( \overline{x_1} + b(y_b).\enc{\nil,y_b.X_b})        \Par 
        x_l.\rec{X_c}( \overline{x_2} +  c(y_c).\enc{\nil,y_c.X_c}) \Par \tree(\cdots) \big)\\
        &=& \nu x_1,x_2.\big(  x_k.\rec{X_b}( \overline{x_1} + b(y_b).y_b.X_b)        \Par 
        x_l.\rec{X_c}( \overline{x_2} + c(y_c).y_c.X_c) \Par \tree(\cdots) \big)
\end{eqnarray*}
The process $P$ can undo either $b[k]$ or $c[l]$. 
This is mimicked by its encoding
since $\enc{P,\phi}$ can undo $b$ on channel $x_k$ or can undo $c$ on channel $x_l$.
Consider
$P \rtran{b[k]} \rtran{c[l]} a[h].(b.\nil \Par c.\nil) = P_a.$
The corresponding transitions of its encoding are as follows:
\[ \enc{P,\phi} \pitran{x_k}\pitran{x_l}  \nu x_1,x_2. \big(\rec{X_b}( \overline{x_1} + b(y_b).y_b.X_b) \Par 
  \rec{X_c}( \overline{x_2} + c(y_c).y_c.X_c) \Par \tree(\cdots)\big) = R_a
\]
$P_a$ can immediately undo $a[h]$ to $a.(b.\nil \Par c.\nil)$, while the two $\mathtt{rec}$ subprocesses of $R_a$ need to first synchronise with $\tree(\cdots)$ on $x_1, x_2$ (see below), before performing $x_h$, which corresponds to $a[h]$, and returning to
$\enc{a.(b.\nil \Par c.\nil)}$. So an atomic transition in CCSK is represented by three transitions in a $\pi$ encoding.
\begin{align*}
 R_a& \pitran{\tau}  \nu x_1,x_2.(\nil \Par \rec{X_c}( \overline{x_2} + c(y_c).y_c.X_c) \Par x_2.x_h.\enc{a.(b.\nil \Par c.\nil)})\\
 &\pitran{\tau}  (\nil \Par \nil \Par x_h.\enc{a.(b.\nil \Par c.\nil)}) \pistr x_h.\enc{a.(b.\nil \Par c.\nil)} \tag*{\lipicsEnd} 
\end{align*}
\end{example}

\subsection{Correctness}

In this subsection we establish the correctness of the weak encoding 
$\enc{-}$ from CCSK into the $\pi$-calculus. The result is stated
as a (weak) mutual simulation property~\cite{Milner71}.

\begin{definition}[CCSK-$\pi$ mutual simulation]\label{def:weak_eq}
A pair of relations $(\rcal_1,\rcal_2)$ between CCSK processes and $\pi$ processes indexed by  a bijection $\phi$ 
from CCSK names and keys to $\pi$ key names is a \emph{CCSK-$\pi$ mutual simulation} iff $(P,R)_{\phi} \in \rcal_1$ implies
  \begin{enumerate}
  	\item if $P \ftran{\Co{a}[k]} P'$ then $R \pitran{\Co{a}(x)} R'$ with $(P',R')_{\phi[(a,k) \mapsto x]} \in \rcal_1$; 
  	\item if $P \ftran{a[k]} P'$ then $R \pitran{ax} R'$ with $(P',R')_{\phi[(a,k) \mapsto x]} \in \rcal_1$;
  	\item if $P \ftran{\tau[k]} P'$ then $R \pitran{\tau}  R'$ with $(P',R')_{\phi}\in \rcal_1$; 
  	\item if $P \rtran{\Co{a}[k]} P'$ then $R \wpitran{\Co{x}} R'$ with $(P',R')_{\phi_{\setminus (a,k)}}\in \rcal_1$ where $\phi(a,k)=x$;
  	\item if $P \rtran{a[k]} P'$ then $R \wpitran{x}  R'$ with $(P',R')_{\phi_{\setminus (a,k)}}\in \rcal_1$ where $\phi(a,k)=x$;
  	\item if $P \rtran{\tau[k]} P'$ then $R \wpitran{\tau} R'$ with $(P',R')_{\phi}\in \rcal_1$;  
  \end{enumerate}
and 
$(R, P)_{\phi} \in \rcal_2$ implies 
  \begin{enumerate}
    \setcounter{enumi}{6}
  	\item if $R \pitran{\Co{a}(x)} R'$ then $P \ftran{\Co{a}[k]} P'$ with $(R',P')_{\phi[(a,k) \mapsto x]} \in \rcal_2$; 
  	\item if $R \pitran{ax} R'$ then $P \ftran{a[k]} P'$ with $(R',P')_{\phi[(a,k) \mapsto x]} \in \rcal_2$;
  	\item if $R \pitran{\Co{x}} R'$ then $P \rtran{\Co{a}[k]} P'$ with $(R',P')_{\phi_{\setminus (a,k)}}\in \rcal_2$ where $\phi(a,k)=x$;
  	\item if $R \pitran{x} R'$ then $P \rtran{a[k]} P'$ with $(R',P')_{\phi_{\setminus (a,k)}}\in \rcal_2$ where $\phi(a,k)=x$;
  	\item if $R \pitran{\tau} R'$ then one of the following happens:
        \begin{enumerate}
        \item $P \ftran{\tau[k]} P'$ with $(R',P')_{\phi}\in \rcal_2$; 
        \item $P \rtran{\tau[k]} P'$ with $(R',P')_{\phi}\in \rcal_2$;
        \item $(R', P)_{\phi}\in \rcal_2$. 
       \end{enumerate}
  \end{enumerate}
We say that a CCSK process P and a $\pi$-calculus process $R$ are mutually similar if there exists a CCSK-$\pi$ mutual simulation $(\rcal_1,\rcal_2)$ with $(P,R)_\phi \in \rcal_1$ and $(R,P)_\phi \in \rcal_2$.
\end{definition}
Note that $(\mathcal{R}_1,\mathcal{R}_2)$ is a pair of simulations: $\rcal_1$ shows how a CCSK process can be simulated by 
a $\pi$ process, and $\rcal_2$ shows how a $\pi$ process can be simulated by a CCSK process.
$\rcal_1$ relates a CCSK process to its canonical encoding, whereas $\rcal_2$ also accounts for intermediate $\tau$-reachable target states
arising from synchronisation with $\tree{(\cdots)}$ (rollback propagation).

\begin{restatable}[Correctness for arbitrary parallel composition]{theorem}{encsmutuallysimilar} \label{th:weak}
Let $P$ be a CCSK process. Then $P$ and $\enc{P, \nil,\varphi(P)}$ are CCSK-$\pi$ mutually similar.
\end{restatable}
Mutual simulation implies both interaction sensitivity and success sensitivity, but the encoding is not compositional; hence this result is not in contrast with the impossibility result in \Cref{sec:separation}.
\Comment{
\begin{proof}

Given the following relations:
\[
  \rcal_1 \;=\; \big\{\, (P, \enc{P, \nil,\varphi(P)}) \mid 
        P \;\text{CCSK process}
        \,\big \}  \qquad 
        \rcal_2 \;=\; \big\{\, (R, P) \mid 
         \enc{P,\nil,\varphi(P)} \wpitran{} R \big\}
\]
we have to show that $\rcal_1$
satisfies conditions $1-6$ and $\rcal_2$ satisfies conditions $7-11$
of Definition ~\ref{def:weak_eq}. Then
$(\rcal_1,\rcal_2)$ is a CCSK-$\pi$ weak mutual simulation as required.
The proof can be found in Appendix \ref{app:weak}.
\end{proof}
Notably, even in the new encoding,  if we restrict attention to the forward direction, 
every CCSK process $P$ is strongly bisimilar to
its encoding $\enc{P}$ in the $\pi$-calculus. 
Dealing with lower-level parallel composition only affects the backward behaviour of the encoding.
}
\begin{restatable}{proposition}{weakparallelpreserving}\label{prp:enc_weak_pp}
The encoding in Figure~\ref{fig:enc_strong} extended in Definition~\ref{def:tree} is parallel preserving.
\end{restatable}

\section{Related work and conclusion}\label{sec:conc}
%
\subparagraph*{Related work}
The expressive power of process calculi has been extensively
studied by means of encodings and separation results, see, e.g.,
\cite{Palamidessi03,Gor10,PT20,PetersN12,PetersY24}. Among the calculi closer to 
 ours,
we recall the following results. Boreale~\cite{Boreale98} gives a compositional encoding of the asynchronous $\pi$-calculus ($\pi_a$)  (without the match operator) into the $\pi$-calculus with internal mobility ($\pi$I)~\cite{Sangiorgi96a},
which is the target of our encodings of CCSK.
However, there exists no valid encoding of $\pi_a$ (with the match operator) into
CCS~\cite[Theorem~5.1]{Gor10}.  (A weaker version of this result was
first
shown in~\cite{HaagensenMP08}.)  
The $\pi$-calculus with implicit matching 
can be encoded in CCS$_\gamma$~\cite{vGl24}, which raises the question of whether CCSK can also be encoded in CCS$_\gamma$.

Encodings or their impossibility have been far less studied for reversible calculi.
In~\cite{LMM19} an LTS isomorphism between CCSK and RCCS (and vice versa) is presented (but it relies on encodings which are not compositional), while in~\cite{MelgrattiMP24}
a truly concurrent semantics of RCCS is given via an encoding in Petri nets. We are only aware of one encoding of a reversible
calculus into a non-reversible one~\cite{LaneseMS16}:
an encoding of a reversible higher-order asynchronous $\pi$-calculus
into its irreversible version (extended with abstractions,
applications, biadic channels and join patterns) is presented, and proved correct w.r.t.\ a form of weak barbed
bisimilarity. This is in line with our results, since the target
calculus in~\cite{LaneseMS16} is more expressive than internal $\pi$
due to the presence of abstractions~\cite{Sangiorgi96a} and join patterns;
hence our impossibility results do not apply. Also, like ours, their
encoding is not compositional, since it takes as additional parameter
a name to be used for rollback.
%
\subparagraph*{Conclusion}
We have studied the encodability of reversible process calculi into forward-only concurrent models, using CCSK as a representative. 
Our results show that reversibility strictly increases expressive power: even with parallel composition restricted to the top level, no basic, success-sensitive encoding into CCS or the $\pi$-calculus exists.

%

We pinpoint the boundary of encodability. When parallel composition is limited to the top level, CCSK can be encoded into the internal $\pi$-calculus up to strong bisimilarity. With arbitrary parallelism, however, no parallel-preserving encoding achieves strong behavioural correspondence; instead, we provide an encoding correct under weak mutual simulation. Taken together, these results clarify the r\^{o}le of reversibility in concurrent computation, showing that it cannot be compiled away without either restricting the source language or weakening the behavioural equivalence. Concerning behavioural equivalence, we conjecture that our last result can be made stronger by using instead of weak mutual simulation a relation inspired by correspondence simulation~\cite{PvG15}. This is by design an \emph{asymmetric} relation, where the second term in order to answer a challenge may go through intermediate steps which have no equivalent in the first term. However, the definition in~\cite{PvG15} is in an unlabelled setting; hence we would first need to extend it to a labelled one. 

Several other directions for future work emerge from this study.
First, our encodability and separation results rely on fundamental features of  reversibility rather than on specific aspects of CCSK. 
This suggests that similar limitations and encodability boundaries may hold for other reversible calculi, such as RCCS~\cite{DK04}, reversible variants of the 
$\pi$-calculus~\cite{CristescuKV13} or higher-order process calculi~\cite{LaneseMS16}. However, this requires detailed analysis, since such calculi separate history information from the actual process; hence, it is not even clear what it means to replace $\nil$ with a process with past behaviour, as we have done in the proof of Proposition~\ref{prop:CCSK not SuRF}. 
Establishing general criteria for the encodability of reversible models remains an open problem.
%
%
%
Second, the encoding for arbitrary parallel composition requires a coordination protocol to ensure that all descendants of a process are reversed before reversing the process itself. Alternative protocols, possibly targeting richer calculi (e.g., with join patterns or higher-order communication), may be more efficient or allow stronger behavioural correspondences.
Third, while our work is purely semantic, reversibility is often motivated by practical applications such as debugging, fault recovery, and reversible programming. Investigating how the theoretical limits identified here affect the design of reversible languages and compilation techniques is a relevant direction for future research.

\subparagraph*{Acknowledgements}
We thank the anonymous referees of CONCUR 2026 for their helpful comments and suggestions.

\bibliography{../../../../bib/axrev-revised}
\newpage
\appendix

\section{Omitted proofs for Section~\ref{sec:separation}}\label{app:separation}

\subsection{Alternative rules for the $\pi$-calculus}\label{subsec:alt pi}

For certain proofs it is more convenient to define a transition relation for the $\pi$-calculus without structural congruence.

\begin{definition}\label{def:alt pi}
Let $\rectran{}$ be the transition relation generated by the rules in Figure~\ref{fig:pi_sem}, including (\textsc{Rec}) but omitting (\textsc{Str}), with the addition of symmetrical rules for (\textsc{Par-L}), (\textsc{Com-L}) and (\textsc{Close-L}).
\end{definition}

\begin{lemma}[{cf.\ \cite[Lemma 1.4.15(1)]{pibook}}]\label{lem:pistr swap}
If $R \pistr \, \rectran\mu R'$ then $R \rectran\mu \, \pistr R'$.
\end{lemma}

\begin{lemma}\label{lem:pitran rectran}
$R \pitran\mu R'$ iff $R \rectran{\mu} \, \pistr R'$.
\end{lemma}
\begin{proof}[Proof sketch]
($\Rightarrow$)
By induction on derivations.
The only extra rule for $\pitran{}$ is \textsc{Str}.
Suppose $R \pitran\mu R'$ comes from $R \pistr \, \pitran\mu \, \pistr R'$.
By inductive hypothesis, $R \pistr \, \rectran\mu \pistr \, \pistr R'$.
By Lemma~\ref{lem:pistr swap}, $R \rectran\mu \, \pistr \, \pistr \, \pistr R'$.
Hence $R \rectran{\mu} \, \pistr R'$.

($\Leftarrow$)
We show that if $R \rectran\mu R'$ then $R \tran\mu R'$.
If rule (\textsc{Par-R}) is used in the derivation of $R \rectran\mu R'$, we replace this with (\textsc{Par-L}) followed by (\textsc{Str}) to swap the order of the parallel composition.
Similarly for (\textsc{Com}) and (\textsc{Close}).
\end{proof}
%

Let $R \recred R'$ iff $R \rectran \tau R$, and let $R \wrecred R'$ iff $R \recred^* R'$.
\begin{lemma}\label{lem:wred wrecred}
$R \wred R'$ iff $R \wrecred \, \pistr R'$.
\end{lemma}
\begin{proof}
By Lemmas~\ref{lem:pistr swap} and~\ref{lem:pitran rectran}.
\end{proof}

Using Lemmas~\ref{lem:pistr swap},~\ref{lem:pitran rectran} and~\ref{lem:wred wrecred} we can show that the notions of strong and weak barbs (Definitions~\ref{def:pi barb} and~\ref{def:calculus}) coincide for $\pitran{}$ and $\rectran{}$.

\subsection{Proof of Lemma~\ref{lem:CCS pi invisible barb}}\label{subsec:pi invisible barb}


\begin{definition}[$\omega$-simulation {\cite[Def.\ 4.1]{PT20}}]\label{def:presim}
Let $P,Q$ range over CCS or $\pi$-calculus processes.
Let $\mu$ range over labels in the transition system of CCS or the $\pi$-calculus as appropriate.  Let ${\wred} = {\tran\tau^*}$, ${\wtran{\hat\mu}} = {\wred\tran\mu\wred}$ if $\mu \neq \tau$, and ${\wtran{\hat\tau}} = {\wred}$.
\begin{enumerate}
\item $Q \presim{0} P$ for all $P,Q$;
\item
For $0 < k < \omega$, $Q \presim{k} P$ iff whenever $Q \tran\mu Q'$ then $\exists P'$ such that $P \wtran{\hat\mu} P'$ and $Q' \presim{k-1} P'$;
\item
$Q \presim\omega P$ iff for all $k < \omega.\ Q \presim{k} P$.
\end{enumerate}
\end{definition}

\begin{lemma}
[barb preservation]\label{lem:CCS pi presim barb}
In CCS and the $\pi$-calculus, if $P \presim{\omega}Q$ and $P \wbarb \alpha$ then $Q \wbarb \alpha$.
\end{lemma}
\begin{proof}
Suppose $P \presim{\omega}Q$ and $P \wred P' \barb \alpha$.
By~\cite[Lem. 4.2]{PT20}, $P' \presim{\omega}Q$.
By~\cite[Prop. 4.1]{PT20}, $Q \wbarb \alpha$.
\end{proof}

%

\begin{lemma}\label{lem:CCS context invisible}
In CCS, for every context $C$, invisible process $I$ and process $P$,
we have $C[I] \presim{\omega}C[P]$.
\end{lemma}
\begin{proof}
This is established in the proof of \cite[Theorem 4.1]{PT20}
\end{proof}
We can deduce Lemma~\ref{lem:CCS pi invisible barb} for CCS from Lemmas~\ref{lem:CCS context invisible} and~\ref{lem:CCS pi presim barb}.

We next turn to the proof of Lemma~\ref{lem:CCS pi invisible barb} for the $\pi$-calculus.
\begin{lemma}
[{\cite[Lemma 4.9]{PT20}}]\label{lem:pi rep context invisible}
In the $\pi$-calculus with replication, for every context $C$, invisible process $I$ and process $R$,
we have $C[I] \presim{\omega}C[R]$.
\end{lemma}
We show Lemma~\ref{lem:pi rep context invisible} in the setting of recursion rather than replication.
\begin{lemma}\label{lem:pi rec context invisible}
In the $\pi$-calculus with recursion, for every context $C$, invisible process $I$ and process $R$,
we have $C[I] \presim{\omega}C[R]$.
\end{lemma}
We can deduce Lemma~\ref{lem:CCS pi invisible barb} for the $\pi$-calculus with recursion from Lemmas~\ref{lem:pi rec context invisible} and~\ref{lem:CCS pi presim barb}.

\begin{lemma}\label{lem:pistr presim}
For all $k \geq 0$, if $R \pistr \presim{k} S$ then $R \presim{k} S$.
\end{lemma}
\begin{proof}
By induction on $k$.
Assume $R \pistr R' \presim{k} S$.
Suppose $R \pitran\mu R''$.  Then $R' \pitran\mu R''$ using (\textsc{Str}).
So $S \wtran{\hat\mu} S''$ with $R'' \presim{k-1} S''$.
\end{proof}
Let $\recpresim{k}$ be $\presim{k}$ defined using $\rectran{}$ instead of $\pitran{}$.

\begin{lemma}\label{lem:recpresim presim}
For all $k \geq 0$, if $R \recpresim{k} S$ then $R \presim{k} S$.
\end{lemma}
\begin{proof}
By induction on $k$.
Suppose $R \pitran\mu R'$.  Then $R \rectran\mu R'' \pistr R'$ by Lemma~\ref{lem:pitran rectran}.
So there is $S''$ such that $S \wrectran{\hat\mu} S''$ and $R'' \recpresim{k-1} S''$.
So $S \wtran{\hat\mu} S''$ and $R' \pistr R'' \presim{k-1} S''$.
By Lemma~\ref{lem:pistr presim}, $R' \presim{k-1} S''$ as required.
\end{proof}

\begin{lemma}\label{lem:pi rectran context invisible}
In the $\pi$-calculus with recursion, for every context $C$, invisible process $I$ and process $R$,
we have $C[I] \recpresim{\omega} C[R]$.
\end{lemma}
Clearly Lemma~\ref{lem:pi rec context invisible} follows from Lemmas~\ref{lem:recpresim presim} and~\ref{lem:pi rectran context invisible}.
The rest of this subsection is devoted to showing Lemma~\ref{lem:pi rectran context invisible} by induction on contexts.

\begin{lemma}[{Cf.\ \cite[Prop.\ 4.3]{PT20}}]\label{lem:presim k omega}
In the $\pi$-calculus with recursion, for every $k$, $R \recpresim{k}S$ and $S \recpresim{\omega} T$ imply $R \recpresim{k} T$.
\end{lemma}

\begin{lemma}[{Cf.\ \cite[Prop.\ 4.4]{PT20}}]
In the $\pi$-calculus with recursion, let $R$ be a process and $I$ be an invisible process. Then $I \recpresim{\omega} R$.
\end{lemma}

\begin{lemma}[{Cf.\ \cite[Prop.\ 4.8]{PT20}}]\label{lem:presim tau out par nu}
Let $R, S, T$ be $\pi$-calculus processes, and $a,n$ be names.
Then, for every $k$, $S \recpresim{k} R$ implies
\begin{enumerate}
\item 
$\tau.S\recpresim{k}\tau.R$
\item
$\out a \tup n.S\recpresim{k} \out a \tup n.R$
\item
$S \Par T\recpresim{k}R \Par T$
\item
$T\Par S\recpresim{k}T \Par R$
\item\label{item:presim res}
$(\nu n)S\recpresim{k}(\nu n)R$
\end{enumerate}
\end{lemma}
The next lemma is needed to deal with recursive processes.
\begin{lemma}[{Cf.\ \cite[Lem.\ 4.4]{PT20}}]\label{lem:presim par}
Let $R_1, R_2, S_1, S_2$ be $\pi$-calculus processes. Then, for every $k$, $S_1\recpresim{k}S_2$ and $R_1\recpresim{k}R_2$ imply $S_1 \Par R_1\recpresim{k}S_2 \Par R_2$.
\end{lemma}

%

\begin{lemma}[{Cf.\ \cite[Lem.\ 4.3]{PT20}}]\label{lem:presim component sum}
Let $S$ and $\sum_I \pi_i.R_i$ be $\pi$-calculus processes.
Then, for every $k$ and every $j \in I$, if $S \recpresim{k} \pi_j.R_j$ then
$S \recpresim{k} \sum_I \pi_i.R_i$.
\end{lemma}

Processes are closed terms with no free recursion variables.  We extend Definition~\ref{def:presim} to allow terms with free variables:
\begin{definition}[{Cf. \cite[Def.\ 4.2]{PT20}}]
Let $E,F$ range over $\pi$-calculus terms with at most $X$ free (i.e.\ $X$ is possibly free but there are no other free recursion variables).
For $0 \leq k \leq \omega$ define $E \recpresim{k} F$ iff for every process $R$ we have $E\{R/X\} \recpresim{k} F\{R/X\}$.
\end{definition}

\begin{lemma}[{Cf.\ \cite[Lem.\ 4.5]{PT20}}]\label{lem:presim rec unfold}
For every term $R$  with at most $X$ free, $R\{\rec X R/ X\} \recpresim{\omega} \rec X R$.
\end{lemma}
\begin{proof}
This is shown in~\cite{PT20} for CCS.
The proof is similar for the $\pi$-calculus with recursion.
\end{proof}


As observed in~\cite{PT20}, in general input prefix does not preserve~$\recpresim\omega$.
Nevertheless the case for input contexts in the proof of Lemma~\ref{lem:pi rectran context invisible} goes through (Lemma~\ref{lem:context input}). 
\begin{lemma}[{Cf.\ \cite[Lem.\ 4.7]{PT20}}]\label{lem:invisible sub}
Let $I$ be an invisible process and $\sigma$ be a substitution. Then $I\sigma$ is an invisible process.
\end{lemma}

\begin{lemma}[{Cf.\ \cite[Lem.\ 4.8]{PT20}}]\label{lem:context sub}
Let $C$ be a context, $R$ a process and $\sigma$ a substitution such that $\bn(C[R]) \inter \n(\sigma) = \emptyset$. Then $C[R]\sigma =C \sigma[R\sigma]$.
\end{lemma}

\begin{lemma}\label{lem:context input}
Let $C$ be a context, $I$ an invisible process and $R$ a process. For every $k$,
if for every name $n$ we have $C[I]\{n/x\} \recpresim{k} C[R]\{n/x\}$ then $a(x).C[I] \recpresim{k} a(x).C[R]$.
\end{lemma}
\begin{proof}
Use Lemmas~\ref{lem:context sub} and~\ref{lem:invisible sub} as in the proof of \cite[Lem. 4.9]{PT20}.
\end{proof}

\begin{lemma}\label{lem:context rec}
Let $C$ be a context with at most $X$ free, $I$ an invisible process and $R$ a process. If for every substitution $\sigma$ we have $C[I]\sigma \recpresim{\omega} C[R]\sigma$ then $\rec X C[I] \recpresim{\omega} \rec X C[R]$.
\end{lemma}
\begin{proof}
We follow the proof of~\cite[Prop.\ 4.7]{PT20} for CCS.
The proof is similar for the $\pi$-calculus with recursion,
using Lemmas~\ref{lem:presim tau out par nu},~\ref{lem:presim par} and~\ref{lem:presim rec unfold}.

We show that for any term $G$ with at most $X$ free, and any substitution $\sigma$,
\[
G\{\rec X C[I]/X\}\sigma \recpresim{k} G\{\rec X C[R]/X\}\sigma
\]
for every $k$ (we need the substitution $\sigma$ to handle the case for input).
The result follows by setting $G = X$ and $\sigma$ the identity substitution.
By induction on~$k$.  The case for $k=0$ is immediate.
Assume true for $k$; we show $G\{\rec X C[I]/X\}\sigma \recpresim{k+1} G\{\rec X C[R]/X\}\sigma$.
If $G\{\rec X C[I]/X\}\sigma \rectran\mu S'$ we must show $G\{\rec X C[R]/X\}\sigma \wrectran{\hat\mu} R'$ with $S' \presim{k} R'$.
We proceed by induction on the depth of the inference of $G\{\rec X C[I]/X\}\sigma \rectran\mu S'$.
There are cases according to the structure of $G$.

$\bullet$ $G = X$.  
Any transition $G\{\rec X C[I]/X\}\sigma = \rec X C[I]\sigma \rectran\mu S'$
must have come via rule (\textsc{Rec}).
So $C[I]\{\rec X C[I]/X\}\sigma \rectran\mu S'$.
By induction on the depth of inference,
$C[I]\{\rec X C[R]/X\}\sigma \wrectran{\hat\mu}R''$ with $S' \recpresim{k} R''$.
Using $C[I]\sigma \recpresim{\omega} C[R]\sigma$ we get
\[
C[R]\{\rec X C[R]/X\}\sigma \wrectran{\hat\mu}R'
\]
with $R'' \recpresim{\omega} R'$.
By Lemma~\ref{lem:presim k omega}, $S' \recpresim{k} R'$ as required.

$\bullet$ $G = a(x).G'$.  
We must use rule (\textsc{In}).
We have
\[
G\{\rec X C[I]/X\}\sigma \rectran{\sigma(a)n} G'\{\rec X C[I]/X\}\sigma\{n/x\}\ .
\]

Also $G\{\rec X C[R]/X\}\sigma \rectran{\sigma(a)n} G'\{\rec X C[R]/X\}\sigma\{n/x\}$.
By induction hypothesis on $k$ we have
\[
G'\{\rec X C[I]/X\}\sigma\{n/x\} \recpresim{k}  G'\{\rec X C[R]/X\}\sigma\{n/x\}\ .
\]

$\bullet$ $G = \Co a \tup n.G'$. We must use rule (\textsc{Out}).
We have
\[
G\{\rec X C[I]/X\}\sigma \rectran{\sigma(\Co a)\tup n} G'\{\rec X C[I]/X\}\sigma\ .
\]

Also $G\{\rec X C[R]/X\}\sigma \rectran{\sigma(\Co a)\tup n} G'\{\rec X C[R]/X\}\sigma$.
By induction hypothesis on $k$ we have
\[
G'\{\rec X C[I]/X\}\sigma \recpresim{k}  G'\{\rec X C[R]/X\}\sigma\ .
\]

$\bullet$ $G = \tau.G'$. We must use rule (\textsc{Tau}).
Similar to the case for $G = \Co a \tup n.G'$.

$\bullet$ $G = \sum_I \pi_i.G_i$. We must use rule (\textsc{Sum}).  We have
\[
\pi_j.G_j\{\rec X C[I]/X\}\sigma \rectran{\mu} S'
\]
for some $j \in I$ by a shorter proof. 
Hence by induction on the depth of inference,
\[
\pi_j.G_j\{\rec X C[R]/X\}\sigma \wrectran{\hat{\mu}} R'
\]
with $S' \recpresim{k} R'$.
Unless $\wrectran{\hat{\mu}}$ is the empty sequence of transitions
\[
G\{\rec X C[R]/X\}\sigma \wrectran{\hat{\mu}} R'
\]
with $S' \recpresim{k} R'$ as required.
If $\wrectran{\hat{\mu}}$ is the empty sequence of transitions we have 
\[
S' \recpresim{k} R' = \pi_j.G_j\{\rec X C[R]/X\}\sigma
\]
In this case
\[
G\{\rec X C[R]/X\}\sigma \wrectran{\hat{\mu}} G\{\rec X C[R]/X\}\sigma
\]
and
\[
S' \recpresim{k} G\{\rec X C[R]/X\}\sigma
\]
by Lemma~\ref{lem:presim component sum} as required.

$\bullet$ $G = \nu a. G'$.
If $G\{\rec X C[I]/X\}\sigma \rectran\mu S'$ is deduced via rule (\textsc{Res}) then we use induction on depth of inference and Lemma~\ref{lem:presim tau out par nu}(\ref{item:presim res}).

Otherwise, $G\{\rec X C[I]/X\}\sigma \rectran\mu S'$ is deduced via rule (\textsc{Open}).  This case is straightforward.

$\bullet$ $G = G_1 \Par G_2$.  There are six possible rules by which
\[
G\{\rec X C[I]/X\}\sigma \rectran\mu S'
\]
can be deduced: (\textsc{Par-L}), (\textsc{Par-R}), (\textsc{Com-L}), (\textsc{Com-R}), (\textsc{Close-L}) and (\textsc{Close-R}).

(\textsc{Par-L}):
Here we must have $S' = S_1 \Par (G_2\{\rec X C[I]/X\}\sigma)$
where $G_1\{\rec X C[I]/X\}\sigma \rectran\mu S_1$.
By induction on depth of inference, $G_1\{\rec X C[R]/X\}\sigma \wrectran{\hat\mu} R_1$ with $S_1 \recpresim{k} R_1$.
By induction on $k$, $G_2\{\rec X C[I]/X\}\sigma \recpresim{k} G_2\{\rec X C[R]/X\}\sigma$.
We have
\[
G\{\rec X C[R]/X\}\sigma \wrectran{\hat\mu} R' = R_1 \Par (G_2\{\rec X C[R]/X\}\sigma)\ ,
\]
with $S' \recpresim{k} R'$ by Lemma~\ref{lem:presim par}.

(\textsc{Com-L}):
Here we must have $\mu = \tau$ and $S' = S_1 \Par S_2$ where
where $G_1\{\rec X C[I]/X\}\sigma \rectran{\Co a\tup v} S_1$ and $G_2\{\rec X C[I]/X\}\sigma \rectran{av} S_2$.
By induction on depth of inference, $G_1\{\rec X C[R]/X\}\sigma \wrectran{\Co a\tup v} R_1$ with $S_1 \recpresim{k} R_1$ and  $G_2\{\rec X C[R]/X\}\sigma \wrectran{av} R_2$ with $S_2 \recpresim{k} R_2$.
We have
\[
G\{\rec X C[R]/X\}\sigma \wrectran{\hat\mu} R' = R_1 \Par R_2 \ ,
\]
with $S' \recpresim{k} R'$ by Lemma~\ref{lem:presim par}.

(\textsc{Close-L}):
Here we must have $\mu = \tau$ and $S' = \nu b.(S_1 \Par S_2)$ where
where $G_1\{\rec X C[I]/X\}\sigma \rectran{\Co a(b)} S_1$ and $G_2\{\rec X C[I]/X\}\sigma \rectran{ab} S_2$.
By induction on depth of inference, $G_1\{\rec X C[R]/X\}\sigma \wrectran{\Co a(v)} R_1$ with $S_1 \recpresim{k} R_1$ and  $G_2\{\rec X C[R]/X\}\sigma \wrectran{av} R_2$ with $S_2 \recpresim{k} R_2$.
We have
\[
G\{\rec X C[R]/X\}\sigma \wrectran{\hat\mu} R' = \nu b.(R_1 \Par R_2) \ ,
\]
with $S' \recpresim{k} R'$ by Lemma~\ref{lem:presim par} and Lemma~\ref{lem:presim tau out par nu}(\ref{item:presim res}).

The symmetrical cases for (\textsc{Par-R}), (\textsc{Com-R}) and (\textsc{Close-R}) are similar to the above and omitted.

$\bullet$ $G = \rec Y G'$. We must use rule (\textsc{Rec}).
Let $G_I = G'\{\rec X C[I]/X\}\sigma$ and $G_R = G'\{\rec X C[R]/X\}\sigma$.
Then $G\{\rec X C[I]/X\}\sigma = \rec Y G_I$ and we have
$G_I\{\rec Y G_I/Y\} \rectran\mu S'$.
By induction on depth of inference, $G_R\{\rec Y G_R/Y\} \rectran{\hat\mu} R'$ with $S' \recpresim{k} R'$.
Unless $\wrectran{\hat{\mu}}$ is the empty sequence of transitions,
we have $G\{\rec X C[R]/X\}\sigma = \rec Y G_R \rectran{\hat\mu} R'$ with $S' \recpresim{k} R'$ as required.
If $\wrectran{\hat{\mu}}$ is the empty sequence of transitions,
we have $R' = G_R\{\rec Y G_R/Y\} \recpresim\omega \rec Y G_R$ by Lemma~\ref{lem:presim rec unfold}.
Then, as required, $S' \recpresim{k} \rec Y G_R$ by Lemma~\ref{lem:presim k omega}.
\end{proof}

All cases for the induction on contexts are now covered, and 
the proof of Lemma~\ref{lem:pi rectran context invisible} is now complete.

To summarise:
\CCSpiinvisiblebarb*
\begin{proof}
We can deduce Lemma~\ref{lem:CCS pi invisible barb} for CCS from Lemmas~\ref{lem:CCS context invisible} and~\ref{lem:CCS pi presim barb}.

We can deduce Lemma~\ref{lem:CCS pi invisible barb} for the $\pi$-calculus with recursion from Lemmas~\ref{lem:pi rec context invisible} and~\ref{lem:CCS pi presim barb}.

Clearly Lemma~\ref{lem:pi rec context invisible} follows from Lemmas~\ref{lem:recpresim presim} and~\ref{lem:pi rectran context invisible}.
Lemma~\ref{lem:pi rectran context invisible} was shown as a series of lemmas by induction on contexts.
\end{proof}


\subsection{Proof of Proposition \ref{prop:CCS pi strongly RF}}
\CCSpistronglyRF*
\begin{proof}
This was shown in~\cite[Thm.\ 4.1]{PT20} for CCS, and in~\cite[Thm.\ 4.2]{PT20} for the $\pi$-calculus with replication.
It is an immediate consequence of Lemma~\ref{lem:CCS pi invisible barb} for the $\pi$-calculus with recursion.
\end{proof}

\subsection{Proof of Proposition \ref{prop:CCSK strongly RF}}

As we did for the $\pi$-calculus in Section~\ref{subsec:alt pi}, in this subsection we use an alternative set of transition rules where we omit structural congruence and rules (\textsc{Str}) and (\textsc{Str}$^\bullet$) from Figure~\ref{fig:ccsk_sem_full}, and add in the symmetric versions of rules (\textsc{Par}), (\textsc{Par}$^\bullet$), (\textsc{Syn}) and (\textsc{Syn}$^\bullet$).
This leaves strong and weak barbs unchanged.
Reachable processes are also unchanged.
This will enable us to show Lemmas~\ref{lem:context invisible} and~\ref{lem:context standard} by structural induction.

\begin{lemma}
In CCSK, a subprocess of a reachable process is also reachable.
\end{lemma}
\begin{proof}
This follows from the characterisation of reachable processes in~\cite[Prop.\ 3.5]{LaneseP21}.
\end{proof}
\begin{lemma}\label{lem:rch fwd-only}
In CCSK, a process $P$ is reachable iff there is a forward-only path from $\tostd P$ to $P$.
\end{lemma}
\begin{proof}
This follows from the Parabolic Lemma~\cite[Lemma\ 5.12]{PU06}; also from~\cite[Cor.\ 3.6]{LaneseP21}.
\end{proof}

\begin{lemma}[cf.\ \cite{PT20}]\label{lem:invisible red}
In CCSK, if $I$ is an invisible process and $I \red I'$ then $I'$ is also invisible.
\end{lemma}
\begin{proof}
Immediate.
\end{proof}

\begin{lemma}\label{lem:context invisible}
In CCSK, let $C[\hol]$ be a context and let $I$ be invisible.
If $C[I] \tran{\mu[k]} P$ then exactly one of the following holds:
\begin{enumerate}
\item\label{item:context}
 $P = C'[I]$ for some context $C'[\hol]$ such that $C[\hol] \tran{\mu[k]} C'[\hol]$;
\item\label{item:invisible}
 $\mu = \tau$ and $P = C[I']$ where $I \tran{\tau[k]} I'$ (and $I'$ is invisible by Lemma~\ref{lem:invisible red}).
\end{enumerate}
\end{lemma}
\begin{proof}
By structural induction on contexts $C[\hol]$.

Base case.
Suppose $I \tran{\mu[k]} P$.  Since $I$ is invisible, $\mu = \tau$ and we are in case (\ref{item:invisible}).

Prefix $\alpha$.
Suppose that $C[\hol] = \alpha.C_1[\hol]$.
Then $C_1[\hol]$ is standard.
If $\alpha.C_1[I] \tran{\mu[k]} P$ then this must be a forward transition $\alpha.C_1[I] \ftran{\alpha[k]} \alpha[k].C_1[I]$.  Let $C'[\hol] = \alpha[k].C_1[\hol]$.
Then $P = C'[I]$ and $C[\hol] \tran{\mu[k]} C'[\hol]$, and we are in case (\ref{item:context}).

Prefix $\alpha[m]$. 
Suppose that $C[\hol] = \alpha[m].C_1[\hol]$.
If $\alpha[m].C_1[I] \tran{\mu[k]} P$ then there are two possibilities:
\begin{enumerate}
\item 
$k = m$, $\mu = \alpha$, $C_1[I]$ is standard and $\alpha[m].C_1[I] \rtran{\alpha[m]} \alpha.C_1[I]$.
Let $C'[\hol] = \alpha.C_1[\hol]$.
Then $P = C'[I]$ and $C[\hol] \rtran{\mu[k]} C'[\hol]$, and we are in case (\ref{item:context}).
\item
$\alpha[m].C_1[I] \tran{\mu[k]} \alpha[m].P_1$ with $C_1[I] \tran{\mu[k]} P_1$.
By the induction hypothesis either
\begin{enumerate}
\item
 $P_1 = C'_1[I]$ for some context $C'_1[\hol]$ such that $C_1[\hol] \tran{\mu[k]} C'_1[\hol]$; or
\item
 $\mu = \tau$ and $P_1 = C_1[I']$ where $I \tran{\tau[k]} I'$.
\end{enumerate}
In the first case we let $C'[\hol] = \alpha[m].C'_1[\hol]$.
Then $P = C'[I]$ and $C[\hol] \rtran{\mu[k]} C'[\hol]$, and we are in case (\ref{item:context}).
In the second case $P = C[I']$, and we are in case (\ref{item:invisible}).
\end{enumerate}

Parallel composition.
Suppose that $C[\hol] = C_1[\hol] \Par Q$.
If $C_1[I] \Par Q \tran{\mu[k]} P$ then we must have used one of rules (\textsc{Par}), (\textsc{Par}$^\bullet$), (\textsc{Syn}) and (\textsc{Syn}$^\bullet$) or their symmetric versions.
\begin{enumerate}
\item
If $Q \tran{\mu[k]} Q'$ with $k \notin \keys{C[I]}$ then let $C'[\hol] = C_1[\hol] \Par Q'$.
We have $P = C'[I]$ and $C[\hol] \tran{\mu[k]} C'[\hol]$ as required.
\item
If $C_1[I] \tran{\mu[k]} P_1$ with $k \notin \keys{Q}$ then by the induction hypothesis either 
\begin{enumerate}
\item $P_1 = C'_1[I]$ for some context $C'_1[\hol]$ such that $C_1[\hol] \tran{\mu[k]} C'_1[\hol]$;
\item $\mu = \tau$ and $P_1 = C_1[I']$ where $I \tran{\tau[k]} I'$ (and $I'$ is invisible by Lemma~\ref{lem:invisible red}).
\end{enumerate}
In the first case let $C'[\hol] = C'_1[\hol] \Par Q'$.
We have $P = C'[I]$ and $C[\hol] \tran{\mu[k]} C'[\hol]$ as required.
In the second case $\mu = \tau$ and $P = C[I']$ where $I \tran{\tau[k]} I'$ as required.
\item
If $\mu = \tau$ and $C_1[I] \tran{\alpha[k]} P_1$, $Q \tran{\Co{\alpha}[k]} Q'$ then by the induction hypothesis
$P_1 = C'_1[I]$ for some context $C'_1[\hol]$ such that $C_1[\hol] \tran{\alpha[k]} C'_1[\hol]$.
Let $C'[\hol] = C'_1[\hol] \Par Q'$.
We have $P = C'[I]$ and $C[\hol] \tran{\mu[k]} C'[\hol]$ as required.
\end{enumerate}
Further cases for parallel composition are symmetric versions of the above.

We omit the cases for sum and restriction, which are straightforward.
\end{proof}

\begin{lemma}\label{lem:context invisible rch hole}
If $I$ is invisible and $C[I]$ is reachable then $C[\hol]$ is reachable.
\end{lemma}
\begin{proof}
Since $C[I]$ is reachable, there is a backward-only computation from  $C[I]$ to
$\tostd{C[I]} = C_0[I_0]$.
Using Lemma~\ref{lem:context invisible} (and Lemma~\ref{lem:invisible red}),
\[
C[I] = C_k[I_k] \rtran{\mu_k[m_k]} C_{k-1}[I_{k-1}] \rtran{\mu_{k-1}[m_{k-1}]} \cdots \rtran{\mu_1[m_{1}]} = C_0[I_0]
\]
where for $i = k,\ldots,1$, either
$C_i[\hol] \rtran{\mu_i[m_i]} C_{i-1}[\hol]$ and $I_i = I_{i-1}$ or
$\mu_i = \tau$,
$C_i[\hol] = C_{i-1}[\hol]$ and $I_i \rtran{\tau[m_i]} I_{i-1}$.
By omitting the steps where only the invisible process moves and $C_i[\hol] = C_{i-1}[\hol]$, we get a backward-only computation from $C[\hol]$ to $C_0[\hol]$, showing that $C[\hol]$ is reachable.
\end{proof}
\begin{lemma}\label{lem:context standard}
In CCSK, let $C[\hol]$ be a context and let $P$ be standard.
If $C[\hol] \tran{\mu[k]} C'[\hol]$ then $C[P] \tran{\mu[k]} C'[P]$.
\end{lemma}
\begin{proof}
By structural induction on contexts $C[\hol]$.
All cases are straightforward.
The base case holds vacuously, since $\hol$ has no transitions.
The most interesting case is when $C[\hol] = a[m].C_1[\hol]$.
If $C[\hol] \rtran{a[m]} a.C_1[\hol]$ then $C_1[\hol]$ must be standard, and also $C[P] \rtran{a[m]} a.C_1[P]$, relying on $P$ being standard.
If $C[\hol] \tran{\mu[k]} C'[\hol]$ (where $C'[\hol] \neq a.C_1[\hol]$) then $C_1[\hol] \tran{\mu[k]} C'_1[\hol]$ (for some $C'_1[\hol]$),
 and we can use the inductive hypothesis to deduce that $C_1[P] \tran{\mu[k]} C'_1[P]$,
 and hence $C[P] \tran{\mu[k]} a[m].C'_1[P] = C'[P]$.
 
In the case for parallel composition, where $C[\hol] = C_1[\hol] \Par Q$, we also rely on $P$ being standard; otherwise, keys in $P$ could disable transitions of $C[\hol]$.
\end{proof}
\begin{lemma}\label{lem:context standard rch}
In CCSK, if $C[\hol]$ is reachable and $P$ is a standard process then $C[P]$ is reachable.
\end{lemma}
\begin{proof}
Suppose that $C[\hol]$ is reachable.  Consider a forward-only path from $\tostd {C[\hol]} = C_0[\hol]$ to $C[\hol]$.
Using Lemma~\ref{lem:context standard} we get a forward-only path from $C_0[P]$ to $C[P]$, showing that $C[P]$ is reachable.
\end{proof}

Let $\hat\alpha$ range over forward ($\alpha$) or backward ($\rev\alpha$) barbs in CCSK.
\begin{lemma}\label{lem:CCSK wbarb invisible std}
In CCSK, let $C[\hol]$ be a context and let $I$ be invisible and $P$ be standard.
Also, let $C[I]$ be reachable.
If $C[I] \wbarb \hat\alpha$ then $C[P] \wbarb \hat\alpha$.
\end{lemma}
\begin{proof}
Consider $C[I] \wred {\barb \hat\alpha}$.  Since $I$ is invisible, $I$ only contributes reductions and does not produce barb $\hat\alpha$.  We can mimic the computation, omitting reductions performed only by $I$, to get $C[P] \wred C'[P] \barb \hat\alpha$.
We need $P$ to be standard, since otherwise reductions made from $C[I]$ might be prevented in the computation from $C[P]$.

We now give the details.
Since $C[I]$ is reachable, so are $C[\hol]$ (by Lemma~\ref{lem:context invisible rch hole}) and $C[P]$ (by Lemma~\ref{lem:context standard rch}).
Suppose $C[I] \wred {\barb \hat\alpha}$.
Using Lemma~\ref{lem:context invisible} (and Lemma~\ref{lem:invisible red}),
\[
C[I] = C_1[I_1] \tran{\tau[m_1]} C_2[I_2] \tran{\tau[m_2]} \cdots \tran{\tau[m_{k-1}]} C_k[I_k] \tran{\alpha[m]} C_{k+1}[I_{k+1}]
\]
where for $i = 1,\ldots,k-1$, either
$C_i[\hol] \tran{\tau[m_i]} C_{i+1}[\hol]$ and $I_i = I_{i+1}$ or
$C_i[\hol] = C_{i+1}[\hol]$ and $I_i \tran{\tau[m_i]} I_{i+1}$.
Since $\alpha \neq \tau$, also by Lemma~\ref{lem:context invisible} we must have $C_k[\hol] \tran{\alpha[m]} C_{k+1}[\hol]$ and $I_k = I_{k+1}$.
By omitting the steps where only the invisible process moves and $C_i[\hol] = C_{i+1}[\hol]$,
we get a computation $C[\hol] \wred C_k[\hol] \tran{\alpha[m]} C_{k+1}[\hol]$.
Using Lemma~\ref{lem:context standard} we get $C[P] \wred C_k[P] \tran{\alpha[m]} C_{k+1}[P]$, showing that $C[P] \wred {\barb \hat\alpha}$ as required.
\end{proof}

\CCSKstronglyRF*
\begin{proof}
Let $I$ be an invisible process, $P$ be a process and $C[\hol]$ a single-hole context.
Let $P_0 = \tostd P$ and $C_0[\hol] = \tostd{C[\hol]}$.
We assume that $C[I]$ and $C[P]$ are reachable.
Then so are $C[\hol]$ (by Lemma~\ref{lem:context invisible rch hole})
and $C[P_0]$ (by Lemma~\ref{lem:context standard rch}).

Suppose $C[I] \wbarb$, say $C[I] \wbarb \hat\alpha$.
Since $C[P]$ is reachable, there is a backward-only computation from $C[P]$ to $\tostd{C[P]} = C_0[P_0]$.
Since $C[P_0]$ is reachable, there is a forward-only computation from $\tostd{C[P_0]} = C_0[P_0]$ to $C[P_0]$.

Combining, we have a backward-only followed by forward-only computation from $C[P]$ to $C[P_0]$ via $C_0[P_0]$.
If there are any visible actions in this computation then $C[P] \wbarb$ as required.
If not then $C[P] \wred C[P_0]$.
Using Lemma~\ref{lem:CCSK wbarb invisible std} we get $C[P_0] \wbarb \hat\alpha$.
Hence $C[P] \wbarb$ as required.
\end{proof}

\subsection{Proof of Proposition \ref{prop:basic success SuRF}}

\basicsuccessSuRF*
\begin{proof}
Suppose for a contradiction that $\enc{\cdot}$ is a basic, success-sensitive encoding from non-strongly SuRF calculus $\calc_1$ to strongly SuRF $\calc_2$.
In $\calc_1$ we must have context $C_1$, invisible $I$ and process $P$ such that $C_1[I] \wbarb \tick$ but not $C_1[P] \wbarb \tick$.
By compositionality there is a $\calc_2$ context $C_2$ such that $\enc{C_1[Q]} = C_2[\enc Q ]$ for all processes~$Q$ in $\calc_1$.
Using success sensitivity in both directions we have $C_2[\enc{I}] \wbarb \tick$ but not $C_2[\enc{P}] \wbarb \tick$.
Using interaction sensitivity (reverse direction) we have that $\enc{I}$ is invisible.  So $\calc_2$ is not strongly SuRF.  Contradiction.
\end{proof}

\subsection{Proof of Proposition \ref{prop:CCS pi strongly SuRF}}

\CCSpistronglySuRF*
\begin{proof}
By Lemma~\ref{lem:CCS pi invisible barb}.
\end{proof}

\subsection{Proof of Theorem \ref{thm:basic succ sens CCSK pi}}

\basicsuccsensCCSKpi*
\begin{proof}
By Propositions~\ref{prop:basic success SuRF},~\ref{prop:CCSK not SuRF},~\ref{prop:CCS pi strongly SuRF}.
\end{proof}

%


\section{Omitted proofs for Section~\ref{sec:encoding}}\label{sec:app}

\subsection{Proof of Theorem~\ref{th:strong_enc}}

We now present a few lemmas needed to prove the correctness of the
encoding.

\begin{lemma}\label{lemma:ancestor}
  If $\enc{P,\nil} \pitran{\mu} \enc{P',\nil}$ then
  $\enc{P,R} \pitran{\mu} \enc{P',R}$.
\end{lemma}
\begin{proof}
  By structural induction on $P$, with a case analysis on the last
  applied rule.
\end{proof}

\begin{lemma}\label{lemma:str}
Let $P$, $P'$ be CCSK processes and $R$, $R'$ be $\pi$-processes. If $P \str P' \bisim R' \pistr R$ then $P \bisim R$.
\end{lemma}
\begin{proof}
  Follows since both $\str$ and $\pistr$ preserve and reflect transitions.
\end{proof}
\strongencbis*
\begin{proof}
  Consider the relation:
 \begin{align*}
    \rcal=&\Big\{\pair{P}{\enc{P,\nil,\varphi(P)}} \, | \, P \textrm{ is a CCSK process, and} \\ 
    & \,\,\varphi \textrm{ is a bijection between pairs (CCSK name, CCSK key) and key name}\Big\}
  \end{align*}

%
  \noindent We first show that the encoding $\enc{P,\nil,\varphi(P)}$ can simulate $P$.
  The proof is by structural induction on $P$, and for each $P$ by
  coinduction. 
  We have a case analysis on the top-level operator in $P$.
  \begin{itemize}
  \item $P=P_1 \Par P_2$: note that the second parameter of the encoding is always $\nil$. If the transition is not a $\tau$ action, then the thesis follows by the inductive hypothesis on the process performing the transition. In case of a $\tau$ action, if the $\tau$ action is from one of the two components then again the thesis follows by the  inductive hypothesis. Otherwise, the CCSK transition is derived by rule \textsc{Syn} with premises $P_1 \ftran{\Co{a}[k]} P'_1$ and $P_2 \ftran{a[k]} P'_2$ (or the other way around, which is analogous). By the inductive hypothesis we have that $\enc{P_1,\nil} \pitran{\Co{a}(b)} \enc{P_1',\nil}$ and $\enc{P_2,\nil} \pitran{\Co{a}(b)} \enc{P_2',\nil}$. By using  the \textsc{Close} rule ($\pi$-calculus) we have $\enc{P_1 \Par P_2,\nil} \pitran{\tau} \nu b.\enc{P_1',\nil} \Par \enc{P_2',\nil} = \enc{P_1' \Par P_2',\nil}$ given that the CCSK transition created a new key $k$ which has one occurrence in $P_1$ and one in $P_2$.
    The case of backward transitions is analogous, noting that in case of a synchronisation the same key is used in both the premises, hence the corresponding $\pi$-calculus name is the same as well since $\varphi$ is a function. Also, the restriction can be garbage collected since a backward synchronisation in CCSK removes a key, and in the encoding removes both the occurrences of the corresponding variable.  
  \item $P = \nu a.P_1$: the thesis follows by the inductive hypothesis. Note that $a$ is a channel in $P$, hence all key names are different from $a$.
  \item $P = \Sigma_i \rho_i.P_i$: let us first assume that all $\rho_i$ are standard, hence $\rho_i=\alpha_i$. Let us consider the case of input first. We have $P \ftran{a_j[k]} a_j[k].P_j + \Sigma_{i \neq j} \alpha_i.P_i$. By definition of the encoding we have 
  \begin{align*}
&  \enc{\Sigma_i \alpha_i.P_i,\nil} = \rec X (\nil + \Sigma_i \encpref{\alpha_i.P_i,X}) = 
 & \rec X (\nil + a_j(y).\enc{P,y.X} + \Sigma_{i \neq j} \encpref{\alpha_i.P_i,X})
  \end{align*}
   
Now by  the early semantics of the $\pi$-calculus we can choose a name $x_k$ for the input obtaining the transition
\begin{align*}
&\rec X (\nil + a_j(y).\enc{P_j,y.X} + \Sigma_{i \neq j} \encpref{\alpha_i.P_i,X}) \pitran{a_j x_k} 
&\enc{P_j,x_k.X}\subst{\enc{\Sigma_i \alpha_i.P_i,R}}{X}.
\end{align*}

We have $\enc{a_j[k].P_j + \Sigma_{i \neq j} \alpha_i.P_i,\nil} = \enc{P_j,x_k.\enc{a_j.Pj+ \Sigma_{i \neq j} \alpha_i.P_i,\nil}}$
since $P_j$ is standard. The thesis follows.

    Let us consider now the case of output. We have $P \ftran{\Co{a_j}[k]} \Co{a_j}[k].P_j + \Sigma_{i \neq j} \alpha_i.P_i$.  By the definition of the encoding we have $\enc{\Sigma_i \alpha_i.P_i,\nil} = \rec X (\nil + \Sigma_i \encpref{\alpha_i.P_i,X}) = \rec X (\nil + \Co{a_j}(y).\enc{P,\Co{y}.X} + \Sigma_{i \neq j} \encpref{\alpha_i.P_i,X})$. 
   
Now by the early semantics we can choose a name $x_k$ for the bound output obtaining a transition
\begin{align*}
&\rec{X} (\nil + \Co{a_j}(y).\enc{P_j,\Co{y}.X} + \Sigma_{i \neq j} \encpref{\alpha_i.P_i,X}) 
&\pitran{\Co{a_j}(x_k)} \enc{P_j,\Co{x_k}.X}\subst{\enc{\Sigma_i \alpha_i.P_i,R}}{X}.
\end{align*}

We have $\enc{\Co{a_j}[k].P_j + \Sigma_{i \neq j} \alpha_i.P_i,\nil} = \enc{P_j,\Co{x_k}.\enc{\Co{a_j}.Pj+ \Sigma_{i \neq j} \alpha_i.P_i,\nil}}$
since $P_j$ is standard. The thesis follows.\\

Next we consider the case where one of the $\rho_i$, say $\rho_1$, is not standard. As a consequence, the other $\rho_i$ 
are standard, and we denote them as $\alpha_i$. We also assume $\rho_1=a[k]$ (the other cases are analogous). Hence, $P = a[k].P_1 + \Sigma_{i \neq 1} \alpha_i.P_i$. 

We first consider the case of transitions from $P_1$.
By the definition of CCSK semantics we have $a[k].P_1 + \Sigma_{i \neq 1} \alpha_i.P_i \ftran{\rho'} a[k].P'_1 + \Sigma_{i \neq 1} \alpha_i.P_i$ with hypothesis $P_1 \ftran{\rho'} P'_1$. 
By definition of the encoding we have $\enc{a[k].P_1 + \Sigma_{i \neq 1} \alpha_i.P_i,\nil} = \enc{P_1,x_k.\enc{a.\tostd{P_1}+\Sigma_{i \neq 1} \alpha_i.P_i,\nil}}$. By the inductive hypothesis we have $\enc{P_1,\nil} \pitran{\pilab} \enc{P'_1,\nil}$,
 where $\pilab$ and $\varphi$ are according to Definition~\ref{def:bisim}.
By Lemma~\ref{lemma:ancestor} we have 
\begin{align*}
&\enc{P_1,x_k.\enc{a.\tostd{P_1}+\Sigma_{i \neq 1} \alpha_i.P_i,\nil}} &\pitran{\pilab} \enc{P_1',x_k.\enc{a.\tostd{P_1}+\Sigma_{i \neq 1} \alpha_i.P_i,\nil}}
\end{align*}
 and the thesis follows.\\
    

We now consider the undo of $a[k]$ (which requires $P_1$ to be standard). By the definition of CCSK semantics
$a[k].P_1 + \Sigma_{i \neq 1} \alpha_i.P_i \rtran{a[k]} a.P_1 + \Sigma_{i \neq 1} \alpha_i.P_i$
we have $$\enc{a[k].P_1 + \Sigma_{i \neq 1} \alpha_i.P_i,\nil} = \enc{P_1,x_k.\enc{a.\tostd{P_1}+\Sigma_{i \neq 1} \alpha_i.P_i,\nil}}$$
Since $P_1$ is standard and contains no parallel composition we have $P_1 = \nu{S}.\Sigma_j \beta_j.P_j$.
Hence,
\begin{align*}
&\enc{P_1,x_k.\enc{a.\tostd{P_1}+\Sigma_{i \neq 1} \alpha_i.P_i,\nil}} = \\
&\nu{S}. \rec X ( x_k.\enc{a.\tostd{P_1}+\Sigma_{i \neq 1} \alpha_i.P_i,\nil} + \Sigma_j \encpref{\beta_j.P_j,X})
\end{align*}
Since $x_k \notin S$ and thanks to $\pi$-calculus semantics we have $\nu{S}. \rec X ( x_k.\enc{a.\tostd{P_1}+\Sigma_{i \neq 1} \alpha_i.P_i,\nil} + \Sigma_j \encpref{\beta_j.P_j,X}) \pitran{x_k} \nu{S}. \enc{a.\tostd{P_1}+\Sigma_{i \neq 1} \alpha_i.P_i,\nil} \pistr \enc{a.P_1+\Sigma_{i \neq 1} \alpha_i.P_i,Q} $ since $P_1$ is standard and names in $S$ do not occur free in the scope of restriction. The thesis follows thanks to Lemma~\ref{lemma:str}.
  \end{itemize}
    We now prove the other simulation:  $P$ can simulate the encoding $\enc{P,\nil,\varphi(P)}$.
    The proof is again by structural induction on the CCSK process $P$, and for each $P$ by coinduction.
    \begin{itemize}
    \item $P=P_1\Par P_2$: as in the other direction, all the cases but synchronisation follow directly by the inductive hypothesis. In the case of synchronisation, we have two matching transitions from the two components. 
      There are two sub-cases, depending on whether the subject of the transitions is a channel name or a key name. In the first case, we can see by inspection on the encoding clauses that the output is always bound, hence synchronisation occurs using the rule \textsc{Close-L} of the  $\pi$-calculus. By the inductive hypothesis we have corresponding forward transitions in CCSK, with fresh keys that we can choose equal, ensuring that $\varphi$ remains a bijection.
Details of the correspondence are as in the other direction.
      In the second case, we can see by inspection on the encoding that the output is always free, hence the \textsc{Com-L}  is applied. By the inductive hypothesis on the two transitions, the steps are mimicked by matching backward transitions since $\varphi$ is a bijection. Notice  that the equality of key names in the $\pi$-calculus ensures the equality of both keys and channel names. Details of the correspondence are as in the other direction.
  \item $P = \nu a.P_1$: the thesis follows by the inductive hypothesis.

  \item $P = \Sigma_i \rho_i.P_i$: by definition of the encoding, we have
    $\enc{\Sigma_i \rho_i.P_i}=\rec X (R + \Sigma_i \encpref{\alpha_i.P_i,X})$
    By
    inspection of the rules of the encoding, $R$ is either $\nil$ or it is guarded by a
    prefix on a key name.  Vice versa, the rest of the translation is
    guarded by prefixes on channel names. If the transition is not
    from $R$, then it is on a channel name and it is matched by a
    forward transition of $P$. Details of the correspondence are as
    for the other direction. 

    If instead the transition is from $R$, it is matched by a backward
    transition of $P$ (indeed, $P$ is non standard if $R$ is not
    $\nil$). By well-formedness of the CCSK process, at most one branch
    is non-standard. The details of the correspondence are as in the other direction.
\qedhere  
  \end{itemize}

\end{proof}

\subsection{Proof of Proposition~\ref{prop:encpar}}
\encisparallelpreserving*

\Comment{
\begin{proposition}\label{prop:encpar}
The encoding in Figure~\ref{fig:enc_strong} is parallel preserving.
\end{proposition}
}
\begin{proof}
The encoding is only defined for processes with parallel composition at the top level. Hence, 
the only context $C[\hol]$, where the encoding is defined for $C[P_1 \Par \cdots \Par P_n]$,
is $\nu \tilde a.(\hol \Par Q)$, where $Q$ is in the domain of the encoding. 
Also, the encoding is homomorphic w.r.t.\ parallel composition and restriction; hence the result.
\end{proof}

\subsection{Proof of Theorem~\ref{thm:par pres strong bis}}

\begin{lemma}\label{lem:pi context barb}
In the $\pi$-calculus, let $C[\hol]$ be a context, and let $R,S$ be processes.
\begin{enumerate}
\item \label{item:nil R}
If $C[\nil] \barb \alpha$ then $C[R] \barb \alpha$;
\item \label{item:R nil R}
if $C[R] \barb \alpha$ then $C[\nil] \barb \alpha$ or $R \barb \alpha$;
\item \label{item:R S S}
if $C[R] \not \barb \alpha$ and $C[S] \barb \alpha$ then $S \barb \alpha$;
\item \label{item:CR S CS}
if $C[R] \barb \alpha$ and $S \barb \alpha$ then $C[S] \barb \alpha$.
\end{enumerate}
\end{lemma}
\begin{proof}
We use the version of the $\pi$-calculus with just rules and no structural congruence, as defined and used in Section~\ref{subsec:alt pi}.
\begin{enumerate}
\item 
By induction on contexts.
All cases are straightforward, apart from the case for $\rec X C[\hol]$.
To handle this, we show by induction on the length of derivation that for all terms $G$ (with at most $X$ free) we have $G\{\rec X C[\nil]/X\}\barb\alpha$ implies $G\{\rec X C[R]/X\}\barb\alpha$.
Cf.\ the proof of Lemma~\ref{lem:pi rectran context invisible} in Section~\ref{subsec:pi invisible barb}.
\item
By induction on contexts.
All cases are straightforward, apart from the case for $\rec X C[\hol]$.
To handle this, we show by induction on the length of derivation that for all terms $G$ (with at most $X$ free) we have $G\{\rec X C[R]/X\}\barb\alpha$ implies $G\{\rec X C[\nil]/X\}\barb\alpha$ or $R\barb\alpha$.
\item
Follows from (\ref{item:nil R}) and (\ref{item:R nil R}).
\item
By induction on contexts.
All cases are straightforward, apart from the case for $\rec X C[\hol]$.
We show by induction on the length of derivation that for all terms $G$ (with at most $X$ free) we have $G\{\rec X C[R]/X\}\barb\alpha$ and $S \barb \alpha$ implies $G\{\rec X C[S]/X\}\barb\alpha$.
\qedhere
\end{enumerate}
\end{proof}

\begin{definition}\label{def:strong bwd refl}
An encoding $\enc\cdot$ from CCSK to the $\pi$-calculus is \emph{strongly backward reflecting} 
if there is a function $\phi$ such that
\begin{enumerate}
\item 
for every visible action $\mu$ and key $k$, $\phi(\mu,k)$ is visible;
\item
for every process $X$, $X \rtran {\mu[k]}$ iff $\enc X \pitran {\phi(\mu,k)}$.
\end{enumerate}
\end{definition}
By Theorem~\ref{th:strong_enc}, the encoding presented in Figure~\ref{fig:enc_strong} is strongly backward reflecting since strong bisimilarity, as in Definition~\ref{def:bisim}, implies the strongly backward reflecting property.

\begin{lemma}\label{lem:star star pres strong bis}
No
parallel-preserving encoding of CCSK into the $\pi$-calculus is strongly backward reflecting.
\end{lemma}
\begin{proof}
We obtain the result by considering the process $a[k].(b \Par c)$ with lower-level parallel composition guarded by an executed action $a$.
Let $C[\hol] = a[k].\hol$.
Then $\enc{a[k].(b \Par c)} = G[R_b \Par R_c]$.
Also $\enc{a[k].(b[m] \Par c)} = G[R'_b \Par R_c]$
and $\enc{a[k].(b \Par c[n])} = G[R_b \Par R'_c]$.
Suppose we have a strong bisimulation as in Definition~\ref{def:bisim}.
Then $a[k].(b \Par c) \rtran{a[k]}$ and $G[R_b \Par R_c] \pitran{\phi(a,k)}$.
Note that the $\pi$-calculus transition must be visible, and so it does not come from a synchronisation.
Similarly we have $G[R'_b \Par R_c] \not\pitran{\phi(a,k)}$ and
$G[R_b \Par R'_c] \not \pitran{\phi(a,k)}$.

Since $G[R'_b \Par R_c] \not\pitran{\phi(a,k)}$ and $G[R_b \Par R_c] \pitran{\phi(a,k)}$ it must be the case that
$R_b \Par R_c \pitran{\phi(a,k)}$ by Lemma~\ref{lem:pi context barb}(\ref{item:R S S}).
So either $R_b \pitran{\phi(a,k)}$ or $R_c \pitran{\phi(a,k)}$ (considering the rules for transitions).
So either $R_b \Par R'_c \pitran{\phi(a,k)}$ or $R'_b \Par R_c \pitran{\phi(a,k)}$.
But then, since $G[R_b \Par R_c] \pitran{\phi(a,k)}$, by Lemma~\ref{lem:pi context barb}(\ref{item:CR S CS}) either $G[R_b \Par R'_c] \pitran{\phi(a,k)}$ or $G[R'_b \Par R_c] \pitran{\phi(a,k)}$, and neither of these hold.  Contradiction.
\end{proof}

\noParPresCCSKtopipressb*

\begin{proof}
By Lemma~\ref{lem:star star pres strong bis}, since strong bisimilarity as in Definition~\ref{def:bisim} implies the strongly backward reflecting property.
%
\end{proof}


\section{Omitted proofs for  Section~\ref{sec:weak_encoding}}\label{app:weak}

\encsmutuallysimilar*
\begin{proof}

We need to show that
\[
  \rcal_1 \;\;=\;\; \{\, (P, \enc{P, \nil,\varphi(P)}) \mid 
        P \;\text{CCSK process}
        \,\}
\]
is a relation satisfying conditions $1-6$ of Definition~\ref{def:weak_eq}, and
\[
  \rcal_2 \;\;=\;\; \{\, (R, P) \mid 
         \enc{P,\nil,\varphi(P)} \wpitran{} R \}.
\]
is a relation satisfying conditions $7-11$ of Definition~\ref{def:weak_eq}. Then
$(\rcal_1,\rcal_2)$ is a CCSK-$\pi$ weak mutual simulation as required.

\emph{Forward direction.}
It is sufficient to observe that for forward computation, that
is clauses (1), (2), (3), (7), (8), and (11.a), 
CCSK-$\pi$ mutual simulation acts as a strong bisimulation.
Also, the encoding of arbitrary parallel composition does not add any $\tau$ step during forward computation, as processes of the form $\tree(\cdots)$ do not have any \emph{forward} behaviour.
Hence 
the result follows by the reasoning used in the proof of Theorem~\ref{th:strong_enc}.

\emph{Backward direction.}
Assume $(P,\enc{P,\nil,\varphi(P)})\in \mathcal{R}_1$. Then
for conditions (4), (5) and (6) the proof is similar to the corresponding cases in the proof of Theorem~\ref{th:strong_enc},
as a backward step of $P$ is directly mimicked by one step of its encoding $\enc{P,\nil,\varphi(P)}$,
since the encoding of a keyed prefix is the same.

Assume $(R,P) \in \mathcal{R}_2$ with $\enc{P,\nil,\varphi(P)} \wpitran{} R$. 
Suppose $R \pitran{\Co{x}} R'$. Since
$R$ is derived from $\enc{P,\nil,\varphi(P)} \wpitran{} R$ 
via just $\tau$ steps, then
there exists an enabled keyed prefix $\Co{a}[k]$ in $P$ such that
$\varphi(a,k) = x$, such that $P\rtran{a[k]} P'$ with $(R',P')_{\varphi\setminus(a,k)}\in \rcal_2$, and this satisfies the condition (9). 
The reasoning for the condition (10) is similar.

Suppose $R \pitran{\tau} R'$. We need to check two conditions: (11b) and (11c), since (11a) has already been dealt with.
Since $R$ makes a $\tau$ move we have two cases: either there exists an enabled
keyed $\tau$-prefix in $P$  or in $P$ there exists a lower-level parallel composition in an active context.
 In the first case we have that $P\rtran{\tau[k]}P'$ and $(R',P')_\varphi\in \rcal_2$: showing (11b). 
 In the second case we have that the $\tau$ move of $R$ is due to some $\tree(\cdots)$ interacting with a prefix $\Co{x}_i $ for some $i$. Hence this step is due to an internal parallel composition, and to a coordinated rollback. Thus,
 this step has no corresponding CCSK transition and is matched by $P$ staying idle, that is $(R',P)_\varphi\in \rcal_2$, which gives (11c).
\end{proof}

\weakparallelpreserving*
\begin{proof}
The result holds for the top-level parallel composition since in this case the encoding is as in Proposition~\ref{prop:encpar}. 
Let us focus on the encoding of a lower-level parallel composition, where the rule from Definition~\ref{def:tree} applies, namely:
\[\enc{\prod_I P_i,R} \,=\, \nu \tilde{x}.\big(\prod_I \enc{P_i,\Co{x_i}} \Par \tree(\tilde{x},R)\big)\]
The required context $G[\bullet]$ is simply $\nu \tilde{x}.\big(\bullet \Par \tree(\tilde{x},R)\big)$, where the choice of $\tilde{x}$ is arbitrary, hence can be fixed. Moreover, $\enc{P_i,\Co{x_i}}$ only depends on $i$, $P_i$ and $C$ (which determines the choice of $x_i$).
\end{proof}

\end{document}